\documentclass[]{interact}

\usepackage[caption=false]{subfig}% Support for small, `sub' figures and tables
\usepackage{amsmath,amssymb,amsthm}
\usepackage{floatflt}
\usepackage{wrapfig}
\usepackage{import}
\usepackage[linesnumbered,vlined,ruled]{algorithm2e}
\usepackage{array}

\usepackage{color}

\usepackage[numbers,sort&compress]{natbib}% Citation support using natbib.sty
\theoremstyle{plain}% Theorem-like structures provided by amsthm.sty
\newtheorem{theorem}{Theorem}[section]

\theoremstyle{definition}

\theoremstyle{remark}

\begin{document}

%\articletype{ARTICLE TEMPLATE}% Specify the article type or omit as appropriate

\title{Move and Time Optimal Arbitrary Pattern Formation by Asynchronous Robots on Infinite Grid}

\author{
\name{Satakshi Ghosh\textsuperscript{a}\thanks{CONTACT Satakshi Ghosh Email: satakshighosh.math.rs@jadavpuruniversity.in} and Pritam Goswami\textsuperscript{a} and Avisek Sharma\textsuperscript{a} and Buddhadeb Sau\textsuperscript{a}}
\affil{\textsuperscript{a}Jadavpur University, Department of Mathematics, Kolkata , West Bengal - 700032, India}
}

\maketitle

\begin{abstract}
% This template is for authors who are preparing a manuscript for a Taylor \& Francis journal using the \LaTeX\ document preparation system and the \texttt{interact} class file, which is available via selected journals' home pages on the Taylor \& Francis website.
The \textsc{Arbitrary Pattern Formation} (\textsc{Apf}) is a widely studied in distributed computing for swarm robots. This problem asks to design a distributed algorithm that allows a team of identical, autonomous mobile robots to form any arbitrary pattern given as input. This paper considers that the robots are operating on a two-dimensional infinite grid. Robots are initially positioned on distinct grid points forming an asymmetric configuration (no two robots have the same snapshot). They operate under a fully asynchronous scheduler and do not have any access to a global coordinate system, but they will align the axes of their local coordinate systems along the grid lines. The previous work dealing with \textsc{Apf} problem solved it in $O(\mathcal{D}^2k)$ robot movements under similar conditions, where $\mathcal{D}$ is the side of the smallest square that can contain both initial and target configuration and, $k$ is the number of robots. Let $\mathcal{D}'=\max\{\mathcal{D},k\}$. This paper presents two algorithms of \textsc{Apf} on an infinite grid. The first algorithm solves the \textsc{Apf} problem using $O(\mathcal{D}')$ asymptotically move optimal. The second algorithm solves the \textsc{Apf} problem in $O(\mathcal{D}')$ epochs, which we show is asymptotically optimal.
\end{abstract}


\section{Introduction}
 In distributed systems robot swarm coordination problems have been studied over the past two decades. The main aim of the distributed system is to use a swarm of inexpensive robots to do any particular work rather than using a very expensive robot. The coordination problem has attracted research interest in recent days.
 \textsc{Arbitrary Pattern Formation} (\textsc{Apf}) is a fundamental coordination problem for autonomous robot swarms. The goal of this problem is to design a distributed algorithm that guides the robots to form any specific but arbitrary pattern given to the robots as an input. In this context, the main research difficulties are, which patterns can be formed, and how they can be formed. In the Euclidean plane, robots can move in any direction with a very small amount of distance, but it is not always possible for robots with weak capabilities to move accurately. So it is interesting to consider this type of problem in grid terrain, where robot movement is restricted in between grid points. In practical applications, the interest has shifted to using a large number of simple robots which are easy to design and deploy and has minimal capabilities. So to make the system cost-effective it is needed to design the robots such that it has minimal capabilities.
 
 In the theoretical framework, depending on the capabilities there are generally four types of robot models. These models are $\mathcal{OBLOT}$, $\mathcal{FSTA}$, $\mathcal{FCOM}$ and $\mathcal{LUMI}$. In each of these models robots are assumed to be autonomous (i.e the robots do not have any central control), identical (i.e the robots are physically indistinguishable), and homogeneous (i.e each robot runs the same algorithm). Furthermore in the $\mathcal{OBLOT}$ model, the robots are silent (i.e there is no means of communication between the robots) and oblivious (i.e the robots do not have any persistent memory to remember their previous state), in $\mathcal{FSTA}$ model the robots are silent but not oblivious, in $\mathcal{FCOM}$ model the robots are oblivious but not silent and in $\mathcal{LUMI}$ model robots are neither silent nor oblivious. The robots do not have access to any global coordinate system. The robots after getting activated operates in a \textsc{Look-Compute-Move} (\textsc{Lcm}) cycle. In \textsc{Look} phase a  robot takes input from its surroundings and then with that input runs the algorithm in \textsc{Compute} phase to get a destination point as an output. The robot then goes to that destination point by moving in the \textsc{Move} phase. The activation of the robots is controlled by a scheduler. There are mainly three types of schedulers considered in the literature. In a synchronous scheduler, time is divided into global rounds. In (\textsc{FSync}) scheduler each robot is activated in all rounds and execute \textsc{Lcm} cycle simultaneously. In a semi-synchronous scheduler (\textsc{SSync})  all robots may not get activated in each round. But the robots that are activated in the same round execute the \textsc{Lcm} cycle simultaneously. Lastly in the asynchronous scheduler (\textsc{ASync}) there is no common notion of time, a robot can be activated at any time. There is no concept of rounds. So there is no assumption regarding synchronization. 
 
 Leader election is an important task for the pattern formation problem, where a unique robot is elected as a leader. In \cite{BoseAKS20} an arbitrary pattern formation algorithm is given in an infinite grid under asynchronous scheduler considering $\mathcal{OBLOT}$ model, but the number of moves is not optimal. So here, this paper aims to form an arbitrary pattern on an infinite grid by a swarm of robots with the optimal number of moves under an asynchronous scheduler considering the $\mathcal{OBLOT}$ model. Furthermore, we propose another algorithm for solving \textsc{Apf} problem on an infinite grid considering asynchronous scheduler and $\mathcal{LUMI}$ model for robots. We will show that the proposed algorithm is time optimal.
 
\section{Related Works} 
The Arbitrary pattern formation problem has been studied in various settings. In the Euclidean plane, this problem was first studied by Suzuki and Yamashita \cite{Suzuki96}. They provided a complete characterization of the class of pattern formable in \textsc{FSync} and \textsc{SSync} for anonymous robots with unbounded memory. Later in \cite{YAMASHITA10} they characterized the families of pattern formable by oblivious robots in \textsc{FSync} and \textsc{SSync}. Then Flochhini \cite{FlocchiniPSW08} studied the cases of solvability of this problem under various assumptions. They showed that without a common coordinate system \textsc{Apf} problem is not solvable, but when there are both axes-agreement the \textsc{Apf} problem can be solved. Further, With one axis agreement, any odd number of robots can form an arbitrary pattern, but an even number of robots cannot in the worst case. In \cite{DieudonnePV10} authors have established  a relationship between \textsc{Leader Election} and \textsc{Arbitrary Pattern Formation} of robots under asynchronous scheduler. Later they also showed that the arbitrary pattern formation is possible to solve when $n \geq 4$ with chirality (resp. $n \geq 5$ without chirality) if and only if leader election is solvable. In \cite{BramasT18} authors consider the arbitrary pattern formation problem with four robots in the asynchronous model with or without chirality.  In \cite{BoseDS21} authors have solved the \textsc{Apf} problem with inaccurate movement, in this case, the formed pattern is very similar to the target pattern, but not exactly the same.  Randomized pattern formation problem was studied in \cite{BramasT16}. In \cite{Cicerone19} authors have shown some configurations where embedded pattern formation is solvable without chirality and some configuration where embedded pattern formation are deterministically unsolvable. The work in \cite{VAIDYANATHAN21} solves \textsc{Apf} in the obstructed visibility model without any agreement in the coordinate system, where they showed that the run time to solve \textsc{Apf} is bounded above by the time required to elect a Leader. A special case of formation problem is mutual visibility problem \cite{AdhikaryBKS18} and gathering problem \cite{StefanoN17}. In mutual visibility, robots are opaque so the main task is to form a configuration where no three robots are co-linear. Recently \textsc{Apf} problem was solved by opaque robots in euclidean plane in \cite{BoseKAS21}. Also, the work in \cite{Bose19} solved this problem with opaque fat robots considering the luminous model. In an infinite grid with opaque robots, the arbitrary pattern formation problem was studied in \cite{kundu22}. Another special case of formation problem, Uniform circle formation is investigated in \cite{FelettiMP18,AdhikaryKS21}. In various graph, for example,  regular tessellation graphs\textsc{Apf} problem was studied in \cite{cicerone20}. Das et al. solved the problem of forming a sequence of patterns in a given order (\cite{FSY15}). Further, they extended the sequence of pattern formation problems for luminous robots (robots with visible external persistent memory) in  \cite{DasFPS20}. There are many works (\cite{FujinagaYOKY15,CiceroneSN19}) where the pattern can be formed by robots with multiplicities. Pattern formation in the presence of faulty robots is an important topic of research. In \cite{PattanayakFMS20} they only allowed crash fault robots.
Another interesting direction of solving this problem is when visibility is limited. Yamauchi in their paper \cite{YamauchiY13}, studied this problem under limited visibility. They first showed that oblivious robots under \textsc{FSync} model and with limited visibility can not solve \textsc{Apf}. Therefore, they considered non-oblivious robots with unlimited memory. For these robots, they presented algorithms that work in \textsc{FSync} with non-rigid movements and in \textsc{SSync} with rigid movements. After that in \cite{LukovszkiH14}, the authors have solved this problem in an infinite grid under 2 hop visibility. The problem was studied in a synchronous setting for robots with constant-size memory, and having a common coordinate system. Bose et. al \cite{BoseAKS20} solved \textsc{Apf} on an infinite grid by oblivious robots with full visibility. They solved this by $O(kD^2)$ move where $k$ is the number of robots and $\mathcal{D}$ is the side of the smallest square which can contain both initial and target configuration.

\section{Problem description and our contribution}
This paper deals with two arbitrary pattern formation problems on an infinite grid. The robots are autonomous, anonymous, identical, and homogeneous. They move only through the edges of the grid. Initially, the robots are placed arbitrarily on the grid. From this configuration, they need to move to a target configuration or, Pattern (a set of target coordinates) without collision. Here we assume that the initial configuration is asymmetric. The robots operate in \textsc{Lcm} cycles under an adversarial asynchronous scheduler. In the first problem, we have considered the $\mathcal{OBLOT}$ model and showed that the \textsc{Apf} can be solved with the optimal number of moves. In the second problem, the $\mathcal{LUMI}$ model is considered where each of these luminous robots has one light which takes three colors and we showed that \textsc{Apf} is solved with optimal time (here time is calculated using the unit epoch where it is assumed that in each epoch a robot has been activated at least once). In the first algorithm, we solved \textsc{Apf} by $O(\mathcal{D}')$ move which is move optimal. And in the second algorithm \textsc{Apf} is solved by $O(\mathcal{D}')$ epoch which is time optimal. Here $D' = \max\{D,k\}$, $k$ is the number of robots, and $\mathcal{D}$ is the side of the smallest square that can contain both initial and target configuration.

\section{Model:}
\paragraph*{Classical oblivious Robots:} In the first problem The $\mathcal{OBLOT}$ model is considered for the robots. In this model, robots are anonymous, identical, and oblivious, i.e. they have no memory of their past rounds. They can not communicate with each other. All robots are initially in distinct positions on the grid. The robots can see the entire grid and all other robots' positions which means they have global visibility. Robots have no access to any common global coordinate system. They have no common notion of chirality or direction. A robot has its local view and it can calculate the positions of other robots with respect to its local coordinate system with the origin at its own position. Here is no agreement on the grid about which line is  $x$ or $y$-axis and also about the positive or negative direction of the axes. As the robots can see the entire grid they will set the axes of their local coordinate systems along the grid lines.

\paragraph*{Robots with lights:} In the second problem the $\mathcal{LUMI}$ model has been considered. In this model, the robots are anonymous and identical and they have constant memory (finite number of lights). Each robot has a light that can assume one color at a time from a constant number of different colors. All the other assumptions are the same as the classical oblivious robots model.

\paragraph*{Look-Compute-Move cycles:}
An active robot operates according to the Look-Compute-Move cycle. In each cycle a robot takes a snapshot of the positions of the other robots according to its own local coordinate system (\textsc{Look}); based on this snapshot, it executes a deterministic algorithm to determine whether to stay put or to move to an adjacent grid point (\textsc{Compute}); and based on the algorithm the robots either remain stationary or makes a move to an adjacent grid point (\textsc{Move}). When the robots are oblivious they have no memory of past configurations and previous actions. After completing each \textsc{Look-Compute-Move} cycle the contents in each robot's local memory are deleted. When each robot is equipped with an externally visible light, which can assume a $O(1)$ number of predefined lights, the robots communicate with each other using these lights.  The lights are not deleted at the end of a cycle. In second algorithm we use one light which takes \textsc{off}, \textsc{head} and \textsc{line} colours.

\paragraph*{Scheduler} 
We assume that robots are controlled by a fully asynchronous adversarial scheduler (\textsc{ASync}). The robots are activated independently and each robot executes its cycles independently. This implies the amount of time spent in \textsc{Look}, \textsc{Compute}, \textsc{Move} and inactive states is finite but unbounded, unpredictable and not same for different robots. The robots have no common notion of time.

\paragraph*{Movement:} In discrete domains movements of robots are assumed to be instantaneous. This implies that the robots are always seen on grid points, not on edges. However, in our work, we do not need this assumption. In the first proposed move optimal algorithm for simplicity, we assume the movements are to be instantaneous, however, this algorithm also works without this assumption, in that case, the robots are asked to wait and do nothing if they see a robot on a grid edge. In the second proposed time-optimal algorithm no such assumption is required. That is, if a robot sees any robot on an edge, then also the robot does its job as directed by the algorithm. The movement of the robots is restricted from one grid point to one of its four neighboring grid points.

\paragraph*{Measuring Run-time:}
Generally, we measure time in rounds in fully synchronous settings.  But as robots can stay inactive for an indeterminate time in semi-synchronous and asynchronous models, here we consider epochs other than rounds. During an epoch, it is assumed that all robots are activated at least once. Here in the second algorithm, we calculate the run-time with respect to epochs.

\section{Optimal move APF algorithm}
\subsection{Agreement on global co-ordinate system}
 Here we consider an infinite grid as a Cartesian product $G=P \times P$. The infinite grid G is embedded in the Cartesian Plane $R^2$. We know that the solvability of Arbitrary pattern formation depends on the symmetries of the initial configuration of the robots. Here we are assuming that the initial configuration is asymmetric. The robots do not have an access to any global coordinate system even though each robot can form a local coordinate system aligning the axes along the grid lines. To form the target pattern the robots need to reach an agreement on a global coordinate system. This subsection provides details of the procedure that allows the robots to reach an agreement on a global coordinate system.

For a given configuration ($\mathcal{C}$) formed by the robots, let smallest enclosing rectangle, denoted by $s.rect(\mathcal{C})$, is the smallest grid aligned rectangle which contains all the robots. Suppose the $s.rect$ of the initial configuration $\mathcal{C}_I$ is a rectangle $\mathcal{R}=ABCD$ of size $m\times n$, such that $m>n>1$. Let $|AB|=n$. Then consider the binary string $\{s_i\}$ associated with a corner $A$, $\lambda_{AB}$ as follows. Scan the grid from $A$ along the side $AB$ to $B$ and sequentially all grid lines of $s.rect(\mathcal{C_I})$ parallel to $AB$ in the same direction. And $s_i=0$, if the position is unoccupied and $s_i=1$ otherwise. Similarly construct the other binary strings $\lambda_{BA}$, $\lambda_{CD}$ and $\lambda_{DC}$. Since the initial configuration is asymmetric we can find a unique lexicographically largest string. If $\lambda_{AB}$ is the lexicographically largest string, then $A$ is called as the leading corner of $\mathcal{R}$.

Next, suppose $\mathcal{R}$ is an $m\times m$ square, then consider the eight binary strings $\lambda_{AB}$, $\lambda_{BA}$, $\lambda_{CD}$, $\lambda_{DC}$, $\lambda_{BC}$, $\lambda_{CB}$, $\lambda_{AD}$, $\lambda_{DA}$. Again since initial configuration is asymmetric, we can find a unique lexicographically largest string among them. Hence we can find a leading corner here as well.

Next, let $\mathcal{C}_I$ be a line $AB$, we will have two strings $\lambda_{AB}$ and $\lambda_{BA}$. Since $\mathcal{C}_I$ is asymmetric then $\lambda_{AB}$ and $\lambda_{BA}$ must be distinct. If $\lambda_{AB}$ is lexicographically larger than $\lambda_{BA}$, then we choose $A$ as the leading corner.

Now for either case, if $\lambda_{AB}$ is the lexicographically largest string then the leading corner $A$ is considered as the origin, and the $x-$ axis is taken along the $AB$ line. If $\mathcal{C}_I$ is not a line then the $y-$ axis is taken along the $AD$ line. If $\lambda_{AB}$ is a line then the $y-$ coordinate of all the positions of robots is going to be zero and in this case, the $y-$ axis will be determined later.   
\begin{figure}[ht]
    \centering
     \includegraphics[width=.7\linewidth]{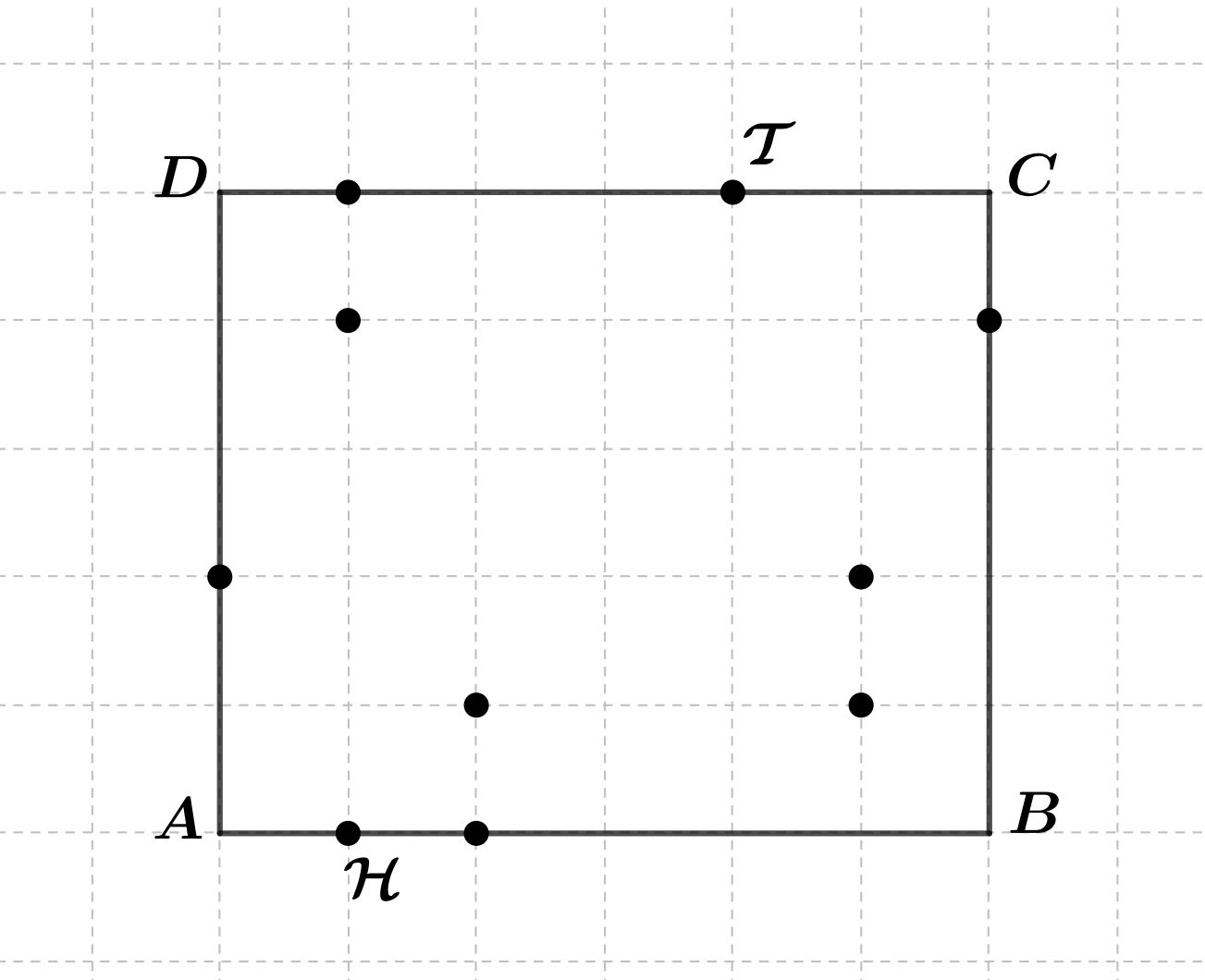}
     \caption{In this configuration AB is the largest string. Here the binary string of AB is 011000000100101000010000000001000010100100. $\mathcal{H}$ and $\mathcal{T}$ are head and tail respectively.}
     \label{Fig:1}
    \end{figure}
For any given asymmetric configuration $\mathcal{C}$ if $\lambda_{AB}$ is the largest associated binary string to $\mathcal{C}$ then the robot causing first non zero entry in $\lambda_{AB}$ is called $head$ let $\mathcal{H}$ and the robot causing last non zero entry in $\lambda_{AB}$ is called as $tail$ let $\mathcal{T}$. We denote the $i^{th}$ robot of the $\lambda_{AB}$ string as $r_{i-1}$. A robot other than head and tail is called \textit{\textbf{Inner robot}}. Further we denote  $\mathcal{C}'=\mathcal{C}\setminus\{tail\}$ and $\mathcal{C}''=\mathcal{C}\setminus\{tail, head\}$ and $\mathcal{C}'''=\mathcal{C}\setminus\{Head\}$,

Let $\mathcal{C}_T$ be the target configuration and $s.rect(\mathcal{C}_T)=\mathcal{R}_T$. Let $\mathcal{R}_T$ is a rectangle of size $M\times N$ with $M\ge N$. We can calculate the binary strings associated with corners in the same manner as previous. $\mathcal{C}_T$ is expressed in the coordinate system with respect to the origin, where origin will be the leading corner. Let the $A'B'C'D'$ be the smallest rectangle enclosing the target pattern with $A'B'\le B'C'$. Let $\lambda_{A'B'}$ be the largest (may not be unique) among all other strings. Then the target pattern is to be formed such that $A=A'$, $A'B'$ direction is along the positive $x$ axis and $A'D'$ direction is along the positive $y$ axis. If target pattern have symmetry then we have to choose any one among the larger string and fixed the coordinate system. So as previously said $head_{target}$ will be the first one and $tail_{target}$ will be the last one in the $s.rect$ of $\mathcal{{C}_T}$. Also we define $\mathcal{C_T}'=\mathcal{C_{T}}\setminus\{tail_{target}\}$, $\mathcal{C_T}''=\mathcal{C_T}\setminus\{head_{target},tail_{target}\}$, $\mathcal{C_T}'''=\mathcal{C_T}\setminus\{head_{target}\}$. We denote the $head_{target}$ position as $t_0$ and $tail_{target}$ position as $t_{k-1}$. Let $H_i$ be the horizontal line having the height $i$ from x-axis. Let for each $i$ there are $p(i)$ target positions on $H_i$. We denote the target positions of $H_0$ as $t_0$, $t_1$......$t_{p(0)-1}$ from left to right. For $H_1$ we denote the target positions as $t_{p(0)}$ to $t_{p(0)+p(1)-1}$ from right to left. For $H_2$ we denote the target positions as $t_{p(0)+p(1)}$ to $t_{p(0)+p(1)+p(2)-1}$ from right to left. Similarly we can denote all other target positions on $H_i$, $i>0$ except $tail_{target}$.

\subsection{Definitions:}
\paragraph*{\textit{\textbf{Inner robot}}:}A robot which is not head or tail in a configuration is called inner robot.
% \paragraph*{\textit{\textbf{Last inner robot}}:} If we order all robots of $\mathcal{C'}$ from head as $r_0$ to $r_{k-2}$ from left to right one by one of all horizontal lines which are parallel to x axis then the $r_{k-2}$ which is last robot in $\mathcal{C'}$  is called \textit{\textbf{Last inner robot}}.
\paragraph*{\textit{\textbf{Compact Line}}:} A line is called compact if there is no empty grid point between two robots.

\subsection{A brief outline of optimal move algorithm}
In this algorithm, our goal is to make \textsc{Apf} with move optimal. For this robots initially form a line and then they form the arbitrary pattern that is given as input. We can divide our algorithm into eight phases. Since the initial configuration is asymmetric the robot can agree on a global coordinate system.  So robots can recognize where the pattern can be formed. Here robots have to maintain the coordinate system during their movements. When a robot wakes up it can be seen in any phase among the eight phases. So as the robots are oblivious they can check by their condition in which phase it is currently in. The conditions are expressed in the Boolean function listed in Table \ref{TABLE-1}. As the movement of the robot is restricted in the discrete domain, here we have to maintain the movement without collision throughout the algorithm. As maintaining the asymmetry is another main challenge of our algorithm, in the first three phases the head will be put at the origin and the tail will expand to the smallest enclosing rectangle for another robot can move but the head and tail remain unchanged and asymmetry also maintained. In phase four robots form a line on x-axis without the tail and one inner robot of $\mathcal{C}'$. In phase five all robots move to the fixed target position sequentially without a head and tail. In the last three phases, the head and tail will reach their fixed target positions. During these phases movement of robots are difficult as here the coordinate system may change. But here we showed that asymmetry and global coordinate will be maintained in all phases. In this way, arbitrary pattern formation can be done.

\vspace{0.01\linewidth}
\begin{table}
\begin{center}
\begin{tabular}{ | p{1.5em} | p{12cm}| } 
   \hline
   $S_0$ & $\mathcal{C}=\mathcal{C}_T$ \\
   \hline
    $S_1$ & $\mathcal{C}'=\mathcal{C}'_T$\\
    \hline
    $S_2$ & x Coordinate of the tail in $\mathcal{C}$ = $x$ Coordinate of $t_{target}$ in $\mathcal{C}_T$ \\
    \hline
    $S_3$ & $m\geq \max{\{N,n\}}+2$ \\
    \hline
     $S_4$ & $m\geq 2.\max{\{M,V\}}$ where $v$ is the length of the vertical side of the smallest enclosing rectangle of $\mathcal{C}'$ \\
    \hline
     $S_5$ & The head in $\mathcal{C}$ is at the origin. \\
    \hline
     $S_6$ & $n\geq \max{\{N+1,H+1,k\}}$ where $H$ is the length of the horizontal side of the smallest enclosing rectangle of $\mathcal{C}'$\\
    \hline
     $S_7$ & $\mathcal{C}''=\mathcal{C}''_T$ \\
    \hline
     $S_8$ & $\mathcal{C}'$ has a non-trivial reflectional symmetry with respect to a vertical line.\\
    \hline
    $S_9$ & $\mathcal{C}'''=\mathcal{C}'''_T$.\\
    \hline
    $S_{10}$ & Line formation on x-axis without tail and one inner robot.\\
    \hline
 
\end{tabular}
\vspace{0.01\linewidth}
\caption{Boolean functions} 
 \label{TABLE-1}
\end{center}
\end{table}
\vspace{0.01\linewidth}

\subsection{Detailed description of the eight phases}

\paragraph*{Phase 1:} A robot is in phase 1 then tail will move upwards and all other robots will remain static. When in phase 1 $\neg{\{S_1\land S_2}\} \land \neg{\{S_3\land S_4}\}$ is true.\\

\begin{theorem}\label{Theorem_1}
If we have an asymmetric configuration $\mathcal{C}$ at some time t
\begin{itemize}
    \item After one move upward, the new configuration is still asymmetric and the coordinate system remains unchanged.
    \item after one move by the tail upwards $\neg{\{S_1\land S_2}\}= true$
\end{itemize}
\end{theorem}

\begin{proof}
Let ABCD be the initial smallest enclosing rectangle at any time t, and let the binary string associated with $AB$ is the largest with $|AB|=n$ and $|AD|=m$, $m \geq n$. So initially the tail is on the side CD. But after the tail moves upward the smallest enclosing rectangle changes. Let the new rectangle is $ABC'D'$. Here tail $\mathcal{T}$ is now on the side $C'D'$.
Now let $\mathcal{T}$ is the only robot initially on side CD , then it is obvious that $\lambda{_{AB}^{new}}$ $>$ $\lambda{_{BA}^{new}}$.
But if there are multiple robots on CD then let $t$ is the $p^{th}$ and $q^{th}$ term of $\lambda{_{AB}^{old}}$ and $\lambda{_{BA}^{old}}$.
Then,\\

\paragraph*{Case-1:}When $p=q$ then $t$ is the middle robot of CD, here $p^{th}$ term is the last 1 occurs in $\lambda{_{AB}^{old}}$ so if we calculate the binary string of first $(p-1)$ term of AB and BA in the new $s.rect$ then also we get $\lambda{_{AB}^{new}}$ $>$ $\lambda{_{BA}^{new}}$.

\paragraph*{Case-2:} When $p>q$ then if we calculate the binary strings of $\lambda{_{AB}^{old}}$ and $\lambda{_{BA}^{old}}$
then $\mathcal{T}$ will be appear earlier in BA, than AB. Now if we calculate the binary strings of first q term of AB and BA then $\lambda{_{BA}^{old}}|_q$ $>$ $\lambda{_{BA}^{new}}|_q$. Also $\lambda{_{AB}^{old}}|_q$ $>$ $\lambda{_{AB}^{new}}|_q$. But $\lambda{_{AB}^{old}}$ $>$ $\lambda{_{BA}^{old}}$. So we have $\lambda{_{AB}^{old}}|_q$ $>$ $\lambda{_{BA}^{old}}|_q$. Finally we can say $\lambda{_{AB}^{new}}|_q$ $>$ $\lambda{_{BA}^{new}}|_q$, so $\lambda{_{AB}^{new}}$ $>$ $\lambda{_{BA}^{new}}$.

\paragraph*{Case-3:} When $p<q$ in that case when $\mathcal{T}$ robot move upward  then in the new ${s.rect}$ we can calculate that in this case also $\lambda{_{AB}^{new}}$ $>$ $\lambda{_{BA}^{new}}$.

So in all the cases $\lambda{_{AB}^{new}}$ $>$ $\lambda{_{BA}^{new}}$. Now we show that the binary string of AB is larger than $C'D'$. As $\mathcal{T}$ is the tail so we know that the binary string of AB is larger than CD, but when the tail robot moves one step upward in that case as there is no robot other than the tail in $C'D'$ so if we calculate binary string it will be $\lambda{_{AB}^{new}}$ $>$ $\lambda{_{C'D'}^{new}}$.\\
In this case, $ABC'D'$ is a non-square grid, so four binary strings to consider here. By calculating we can say that AB is the largest binary string in this new smallest enclosing rectangle. So the new configuration is still asymmetric. So the coordinate system is unchanged. As the tail moves upward so the x coordinate remains unchanged, so $\neg{\{S_1\land S_2}\}$ remains the same after one move by the tail robot.
\end{proof}

\paragraph*{Phase 2:} In this phase head $\mathcal{H}$ will move left towards the origin. When the algorithm is in phase 2 then either $S_3 \land S_4 \land \neg{S_5} \land \neg{S_7}$ or $\neg S_2 \land S_3 \land S_4 \land \neg S_5 \land S_7$ is true.\\

\begin{theorem}
If we have an asymmetric configuration $\mathcal{C}$ at time t in phase 2 then 
\begin{enumerate}
    \item after one move by the head robot to the left, the new configuration is still asymmetric and the coordination system is not changed.
    \item after a finite number of moves by the head to the left $S_5$ is true and phase 2 terminates.
\end{enumerate}
\end{theorem}

\begin{proof}
As $ABCD$ is the smallest enclosing rectangle of all robots, let $\lambda_{AB}$ be the lexicographically largest string, after one move of the head to the left, now also $\lambda_{AB}$ is the largest string. Let at the $i^{th}$ term the first 1 occurs in $\lambda_{AB}$, then in the other strings all the $(i-1)$ terms are 0. But when the head moves one distance to the left then the new string form is now the largest. So the new configuration is asymmetric and the global coordinate will not change. So (1) is true, similarly by the finite number of moves head move to the origin then $S_5$ true.
\end{proof}

\paragraph*{Phase 3:} The goal of this phase is to make $S_6$ true. In this phase the robots will check either $S_8$ is true or false. When algorithm is in phase 3 then $S_3 \land S_4 \land S_5 \land \neg S_6 \land \neg S_7$ = true. When $S_8$ is false, then the tail will move rightwards and the rest will remain static. But when $S_8$ is true, the $\mathcal{C'}$ has a nontrivial reflectional symmetry with respect to a vertical line $V$.

Let the smallest enclosing rectangle is $\mathcal{R}$ = $ABCD$ where $\lvert{AB}\rvert$ = $n$ and $\lvert{AD}\rvert$ = $m$ , $m>n$. Let $\lambda_{AB}$ be the lexicographically largest string. In this case tail will move right and after finite number of move we have $S_3 \land S_4 \land S_5 \land S_6 \land \neg S_7$ = true.\\

\begin{theorem}
If we have an asymmetric configuration $\mathcal{C}$ at some time $t$ then 
\begin{enumerate}
    \item after one move rightward the new configuration is still asymmetric and the global coordinate system remain unchanged.
    \item after one move $S_4 \land S_5 \land \neg S_7$ = true.
    \item after finite number of moves by the tail to the rightwards $S_3 \land S_4 \land S_5 \land  S_6 \land \neg S_7$ = true.
    
\end{enumerate}
\end{theorem}

\begin{proof}
Let the smallest enclosing circle at time $t$ is $ABCD$, where $\lambda_{AB}$ is the largest string. After one move by the tail rightward there may arise two cases.

\paragraph*{Case-1:} Suppose now tail robot is at C, then by one move of tail the new $s.rect$ is $APQD$, where $|AP|$= $(n+1)$. Now it is easy to check that as $m\ge n+2$ so $m>n+1$. So we get that $AD>AP$. This implies that the new configuration is still not square, so we have to consider here only four binary strings, and as earlier $\lambda_{AP}$ will be larger string. So we can conclude that the configuration is still asymmetric and the coordinate system is not changed by one move of the tail. It is easy to check that $S_4$ and $S_5$ are true here but not $S_7$. After the movement of the tail, $S_3$ may become false, so we are in phase 1 then. Then the tail moves upwards and one upwards move still has $S_3 \land S_4 \land S_5 \land  S_6 \land \neg S_7$ is true.\\
\paragraph*{Case-2:} Let after one move by the tail the smallest enclosing circle remain unchanged. As in this phase, the head is in origin and the tail has moved until $S_4$ true, in the binary string of $CD$ or $DC$ is smaller than $AB$. Also in this phase, $S_8$ is not true. So we must have $AB$ larger string than $BA$, so finally we get $\lambda{_{AB}^{new}}$ $>$ $\lambda{_{BA}^{new}}$.
 Note that $S_8$ is either true or false in phase 3 by a finite number of moves of the tail the configuration remains asymmetric.

If $S_8$ is true then there may happen two cases:
 
     \begin{figure}[ht]
     \centering
     \includegraphics[width=.45\linewidth]{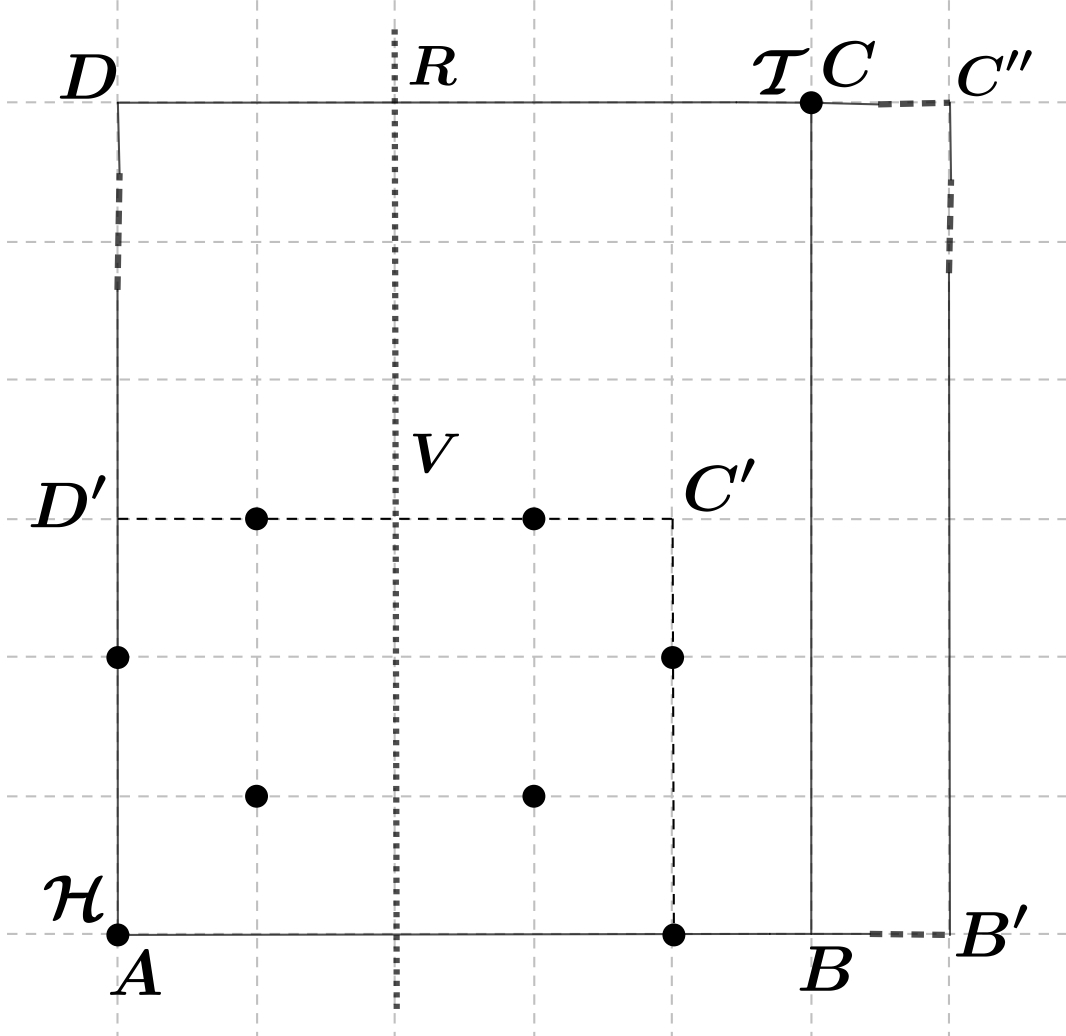}
     \caption{Case-1: $\mathcal{C'}$ has a vertical symmetry and tail will move rightwards }
     \label{Fig:2}
     \end{figure}

    \begin{figure}[ht]
    \centering
    \includegraphics[width=.5\linewidth]{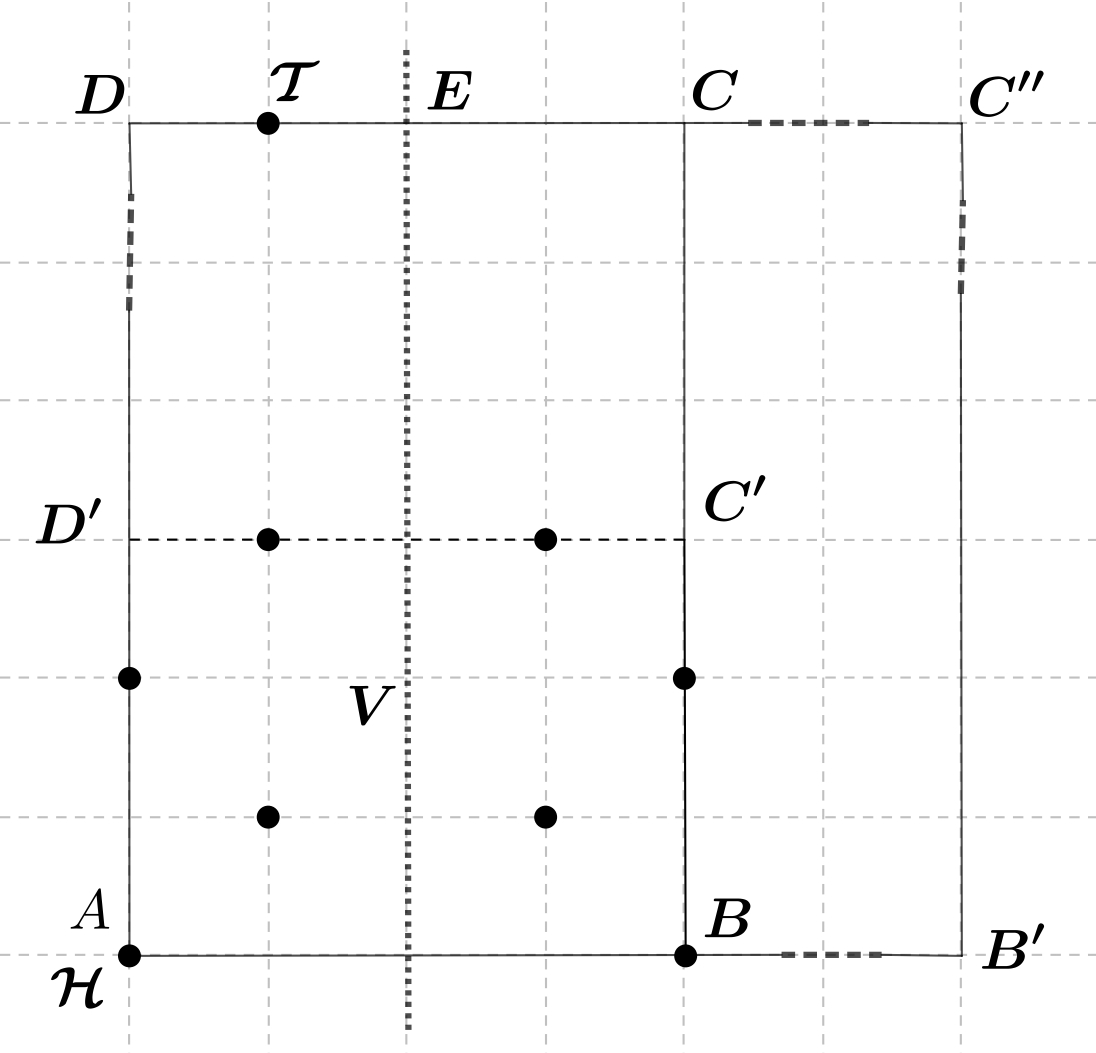}
    \caption{Case-2: $\mathcal{C'}$ has a vertical symmetry and tail will move leftwards }
    \label{Fig:2.a}
    \end{figure}
    
\paragraph*{Case-1:} there is a vertical symmetry in $\mathcal{C'}$ but if we consider $ABCD$ then the tail is rightward of vertical symmetry line $V$. Then tail will move rightwards and after a finite number of moves, $S_6$ is true.
\paragraph*{Case-2:} When the tail is in the left portion with respect to the vertical symmetric line $V$. Then tail will not move rightwards. It will move to the left one more step than $D$, then the coordinate system will be changed. $B$ will be then origin and x-axis = $BA$ and y-axis = $BC$. Then the case will be the same as case-1.

\end{proof}

\paragraph*{Phase 4:} In this phase the  configuration satisfies $S_3 \land S_4 \land S_5 \land S_6 \land \neg S_7 \land \neg S_{10}$ = true. Other than the tail and one inner robot, all other inner robots form a line on the x-axis. In phase four, the head is in origin, and all the robots on the x-axis first make the line compact, i.e. there is no empty grid point between two robots. After the x-axis become compact, when a robot $r_i$ is on $H_i$ and there are no robots in between $H_i$ and the x-axis and the right part of $r_i$ is empty in its horizontal line, then the robot moves to the x-axis. This procedure is done one by one by robots. In between this movement, no collision will occur. Finally when a robot sees that except for itself and tail all other inner robots are on the x-axis then it will not move to the x-axis. So all the inner robots other than the tail and one inner robot form a line on the x-axis.\\

 %\SetKwCommprocedure{Comment}{/* }{ */}
%  \begin{algorithm}[H]
%  \caption{lineFormation}\label{alg:main}
%   %\KwData{}
%  %\KwResult{Formation of configuration II}
%   \If{I am not head and tail and \textit{Last inner robot}}
%   {
%          \If{I am on $x$-axis}
%          {
%           \If{my left position is not occupied}
%               {
%                      move to left unoccupied position\;
%               } 
%          }

%       { ~\hfill/* robot is at $H_i$ such that $i>0$ */ \\
%       \If{there is no other robot in between $H_i$ and $H_0$}
%          {
%              \If{I am the rightmost robot on my line}
%              {
%               move to the first unoccupied position on $H_0$  \;
%              }
%          }  
%      }
  %  }
 %\end{algorithm}
 
\begin{theorem}
If we have an asymmetric configuration $\mathcal{C}$ such that $S_3 \land S_4 \land S_5 \land S_6 \land \neg S_7 \neg S_{10}$ = true  then in phase 4 after finite number of moves by the robots $S_{10}$ becomes true and phase 4 terminates.
\end{theorem}

 \begin{figure}[ht]
    \centering
     \includegraphics[width=.55\linewidth]{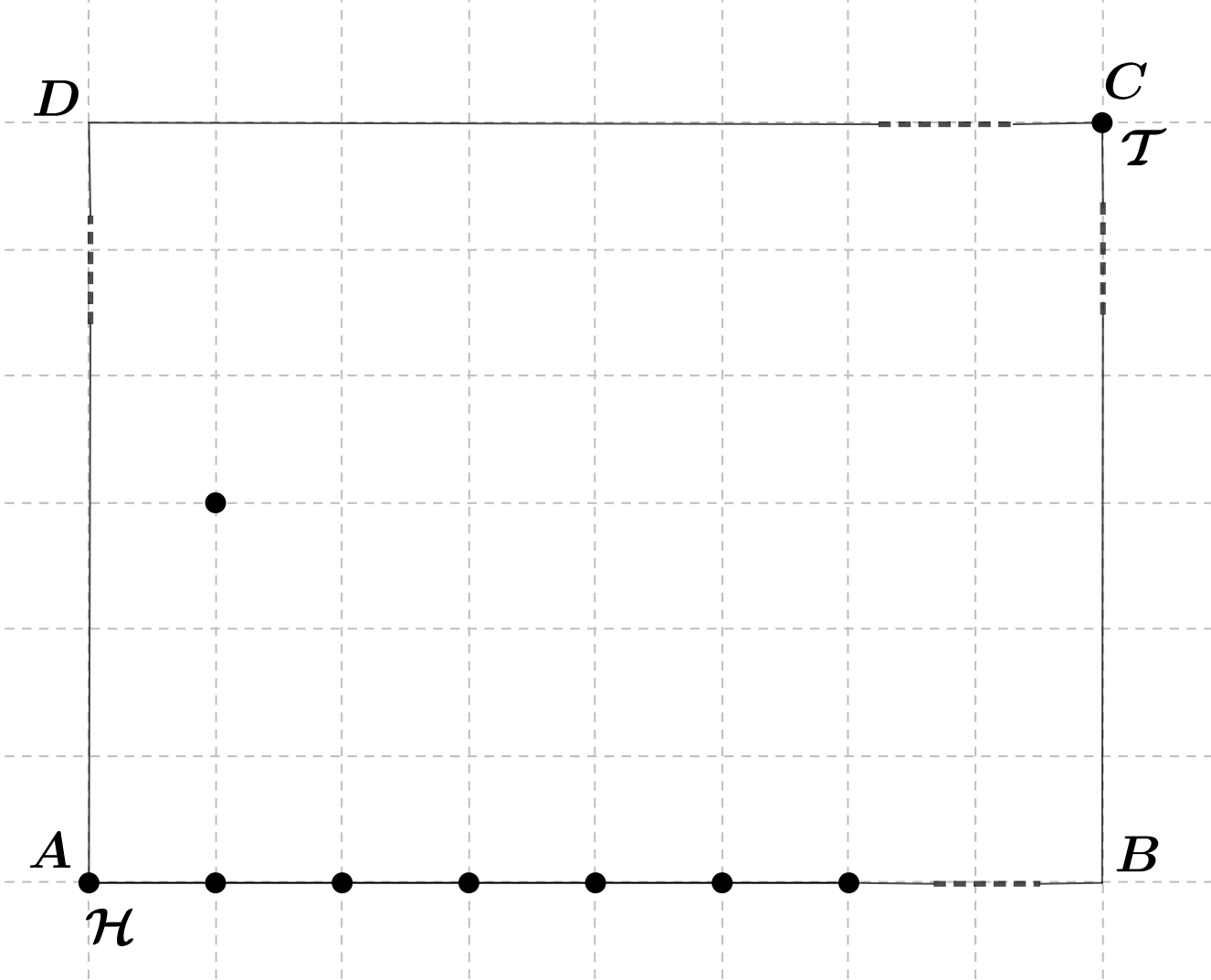}
     \caption{Line formation of robots on x axis without tail and last inner robot. }
     \label{Fig:4}
    \end{figure}

\begin{proof}
In previous phases when the tail robot expands the smallest enclosing rectangle and the head robot moves to the origin, then in this phase robots that are on the x-axis make the line compact. In this move, a robot will move to its left grid point if it is empty, so the robot's movement is in the left direction. So collision will not occur. A robot $r_i$ which is on a horizontal line (let $H_k$) first checks that all the down lines robots are on the x-axis or not, if yes then when a robot sees that the right side of its horizontal line has also no robots, it will move to the x-axis. In this way, robots move down to the x-axis one by one. So the movement of the robot is sequential here. A robot will not move until it sees that it is the rightmost robot in its horizontal line and there is no more robot in the down horizontal lines other than the x-axis. As the ${s.rect}$ is expanded by the tail robot in the previous phases, a robot always gets a path in the grid and moves to the x-axis. So collision will not occur in this movement. Finally without the tail and one inner robot, after a time all other robots will form a line on the x-axis. Hence $S_{10}$ is true. 
\end{proof}

\paragraph*{Phase 5:} In this phase, inner robots will move to the fixed target one by one. When one inner robot sees that all other robots without itself and the tail are on the x-axis then it moves to $t_{k-2}$. We call this inner robot as \textit{Last inner robot}. When a robot on the x-axis sees that the \textit{Last inner robot} is at its target position then $i^{th}$ robot from the left on the x-axis will reach to $t_i$ target position when it sees that from $t_{i+1}$ to $t_{k-2}$ positions are occupied. \\
% all robots are on the x-axis and , then all inner robots $r_i$ move to $t_i$. An inner robot $r_i$ on the x axis moves to it's target position when it sees that $r_{i+1}$ robot is at $t_{i+1}$. So when \textit{Last inner robot} reaches its fixed target position, then the robots on the x-axis from right to left move to the target position one by one. In this case, robots first fill a target position of a horizontal line from left to right and from up to down of the horizontal lines. But when all the target position is filled by robots other than the x-axis, then robots on the x-axis fill the right target position first and then right to left one by one. So in the x-axis right position robot moves to the right target position first. So after this phase, all the robots without the head and the tail are in target positions.\\

 \begin{figure}[ht]
    \centering
     \includegraphics[width=.6\linewidth]{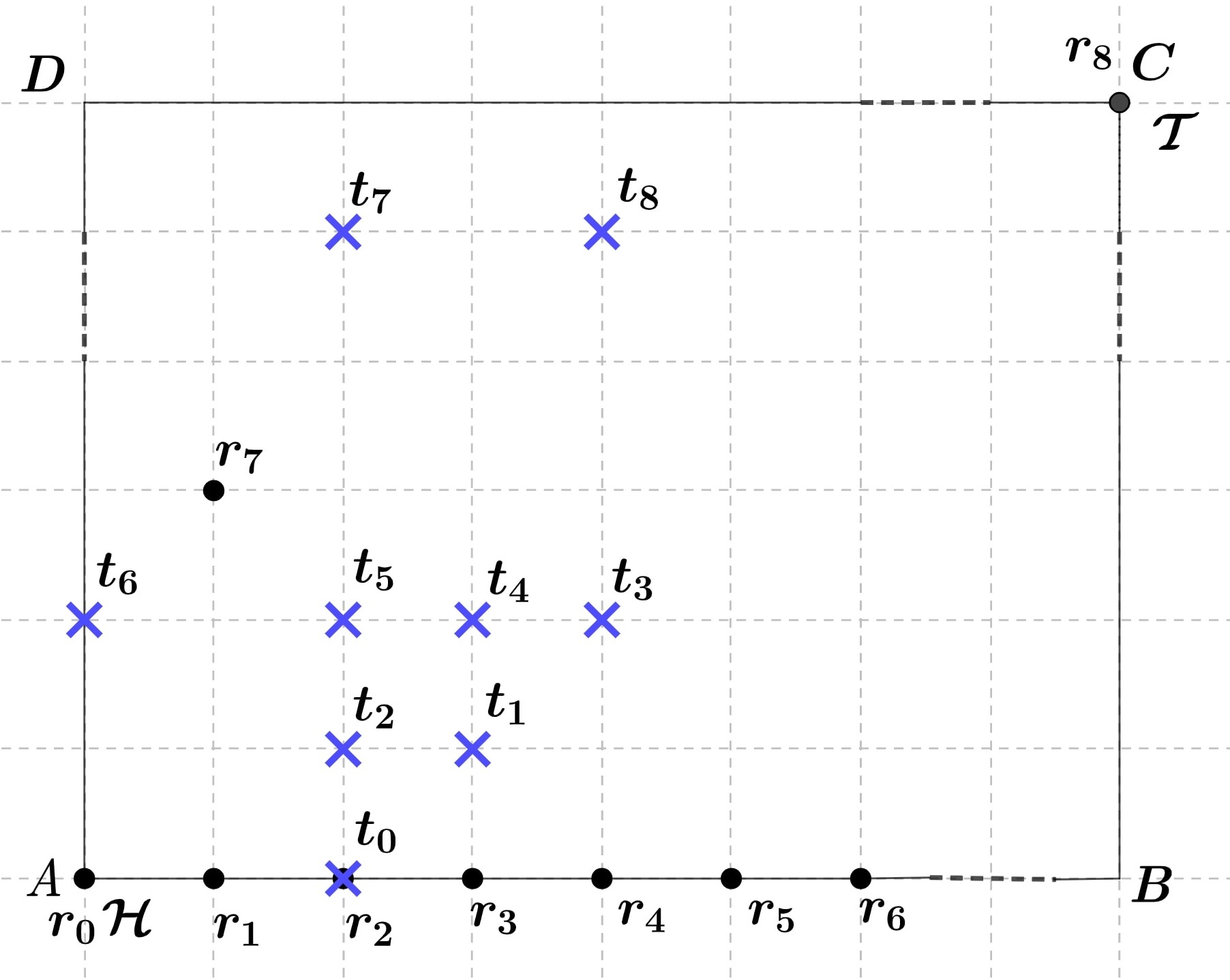}
     \caption{Target formation without head and tail. }
     \label{Fig:5}
    \end{figure}

\begin{theorem}
If we have an asymmetric configuration initially at some time t then
\begin{enumerate}
    \item By movement of inner robots the new configuration is still asymmetric and the coordinate system remains the same.
    \item After any move of the inner robots in this phase $\mathcal{C}''=\mathcal{C}''_T$ and phase 5 terminate and $S_7$ becomes true.
\end{enumerate}
\end{theorem}

\begin{proof}
As the head is in origin and tail robot expands the ${s.rect}$ of the initial configuration, so by the movement of inner robots coordinate system will not change and configuration remains asymmetric. Here our main concern is collision. In this phase an inner robot $r_i$ moves to target position $t_i$ when $t_{i+1}$ to $t_{k-2}$ positions are occupied by robots. As of last inner robot moves to $t_{k-2}$ at the start of this phase. Let us denote the robots on the x-axis from left to right as $r_0,r_1,\dots,r_{k-3}$. Then firstly $r_{k-3}$ reaches $t_{k-3}$, then $r_{k-4}$ reaches $t_{k-4}$ and so on. Finally, $r_1$ moves to $t_1$. As the movement in this phase is sequential that is no other robot moves until one moving robot reaches its target position. Also since the target positions are ordered in such a way that every inner robot will find a  unblock path to reach its target position. So here no collision will occur. Eventually, each inner robot reaches its target position. So finally $S_7$ is true.

% Last inner robot when sees that without the tail all the inner robots are on the x-axis it will order itself as $r_{k-2}$ as the total number of robots k. It will move to $t_{k-2}$. Let $ABC'D'$ be the smallest enclosing rectangle of $\mathcal{C'}$ then as all other inner robots are now on the x-axis and an inner robot $r_i$ will not start its movement until it sees that $r_(i+1)$ has reached it's destination. So in the movement of any $r_{i-1}$ there will be no collision. After the last inner robot reaches its target position, robots on line order themselves, and from right to left they move upward and reach their target position which is ordered from right to left without the tail. The last position of the target pattern will be ordered as $t_{k-1}$. But the others are ordered from left to right in a horizontal line (when not on the x-axis). As a robot moves upward and reaches its target position, then the next robot's target position will be on the right of the previous robot if in the same horizontal line. So one by one all robots reach their position without collision. When all robots' positions other than the x-axis are filled, then robots on the x-axis fill the target position in their line from right to left. So in all the moves by any robot collision will not occur. After this phase without a head and tail, all robots reach their final position. So finally $S_7$ is true.
\end{proof}

\paragraph*{Phase 6:}  The tail will move to the left until the x-coordinate matches with the $tail_{target}$. In this phase, the tail will move to make $S_2$ true.\\

\begin{figure}[ht]
    \centering
     \includegraphics[width=.5\linewidth]{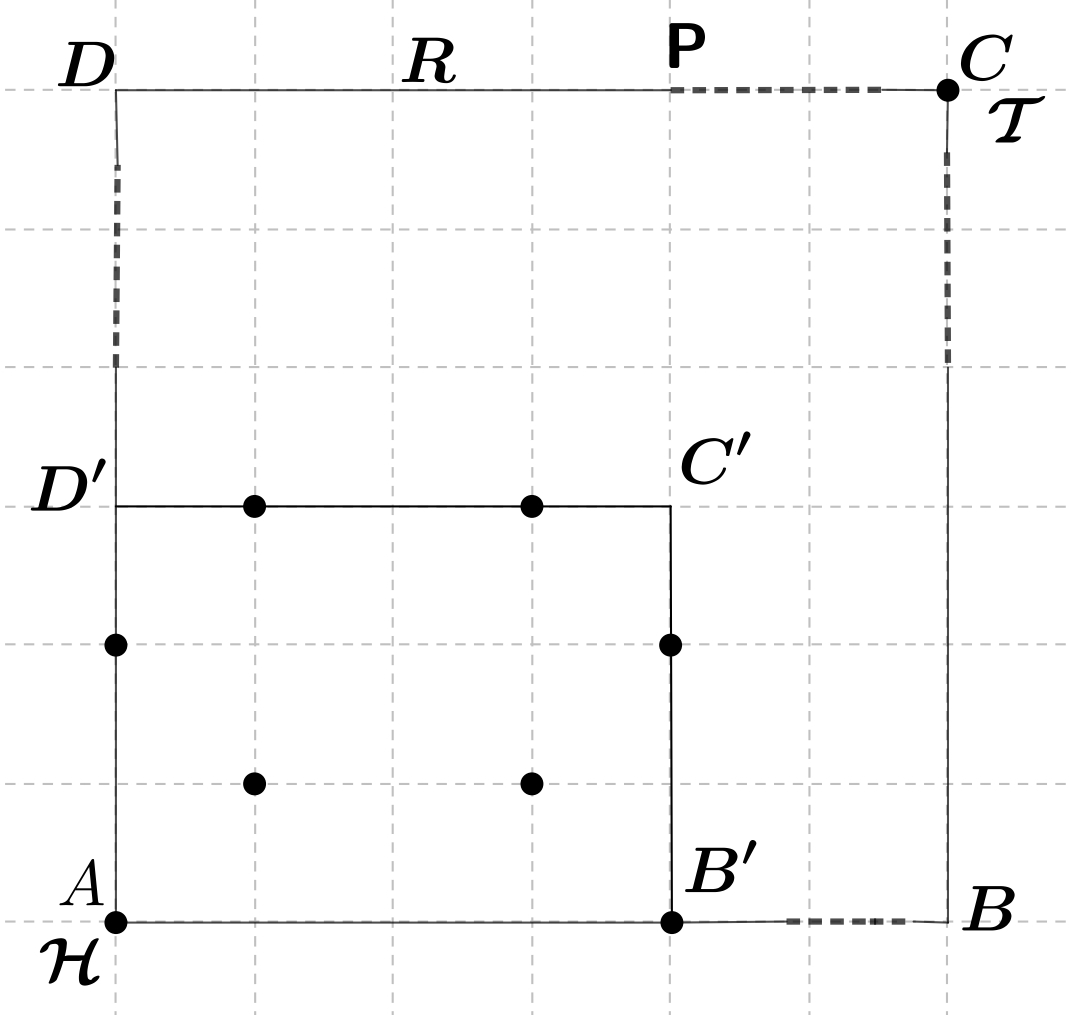}
     \caption{When $S_8$ is true in phase 6}
     \label{Fig:6}
    \end{figure}
\begin{theorem}
If we have an asymmetric configuration $\mathcal{C}$ at some time t then after a finite number of moves phase 6 terminates and $S_2$ becomes true.
\end{theorem}

\begin{proof}
In phase 6 if we have an asymmetric configuration, then depending on whether $S_8$ is true or not there are two cases. Let the smallest enclosing rectangle of $\mathcal{C}$ be ABCD where A is the origin and AB = x-axis and AD = y-axis. Let without tail the smallest enclosing rectangle is $AB'C'D'$. In this phase without tail, all the robots are now on $AB'C'D'$.\\

\paragraph*{Case-1:} Let $S_8$ is not true i.e. there is no symmetry in $\mathcal{C'}$. Then the tail robot will move on the CD side, no matter where $tail_{target}$ is AB will be lexicographically largest string. Finally tail will move up to $S_2$ will true.

\paragraph*{Case-2:} Let $S_8$ is true. Since there is symmetry in $\mathcal{C'}$, here the head robot is in A. Let $P$ be the point of intersection between $B'C'$ and CD. In this case when the tail robot moves left in the line CD, then it will move up to that point whose x coordinate is the same as $tail_{target}$. In this move by the tail when it will reach a point let $R$ which is in between $P$ and D then a vertical symmetry will be created. The tail robot's destination can not be [$P$,$R$] because then $B'$ will be the head. So when the tail crosses $R$ there will be a vertical symmetry. In this case also $S_1$ hold. When the tail robot moves left in CD then in other cases coordinate system remains invariant and $S_2$ holds.\\
So in both cases when phase 6 terminates then $S_2 \land S_3 \land S_4 \land S_5 \land S_7$ true.
\end{proof}

\paragraph*{Phase 7:} The aim of this phase is that Head will moves to $head_{target}$. Consider a configuration that is asymmetric and in phase 7, then with respect to the global coordinate system as fixed, let $ABCD$ is the smallest enclosing rectangle and $\lambda_{AB}$ be the lexicographically largest string. Clearly head is on the side $AB$ and the tail is on $CD$. If we mark all the target points on the grid then let the smallest rectangle be $AB'C'D$. Let $head_{target}$ be the  final position of head, then by finite move head will move to $head_{target}$.\\

\begin{theorem}
Let $\mathcal{C}$ be the asymmetric configuration at time t in phase 7 then by a finite number of moves by the head to the right phase 7 terminate when $\neg S_0 \land S_1 \land S_2 \land S_9$ = true.
\end{theorem}

\begin{proof}
Let ABCD be the smallest enclosing circle of an asymmetric configuration $\mathcal{C}$ of phase 7, Here AB is the lexicographically largest string. Now we have to plot the smallest enclosing rectangle of the target pattern with respect to our current coordinate system. This phase aims to move the head robot from the origin A to its fixed target position, which will be in the right direction of A on AB. As there may be vertical symmetry in target configuration when the head moves to its target then also vertical symmetry will happen. So in all the cases, AB will be lexicographically larger string when the head robot moves to its target position. So when the head reaches its final position phase 7 terminates and $\neg S_0 \land S_1 \land S_2 \land S_9$ is true.
\end{proof}

\begin{figure}[ht]
    \centering
     \includegraphics[width=.5\linewidth]{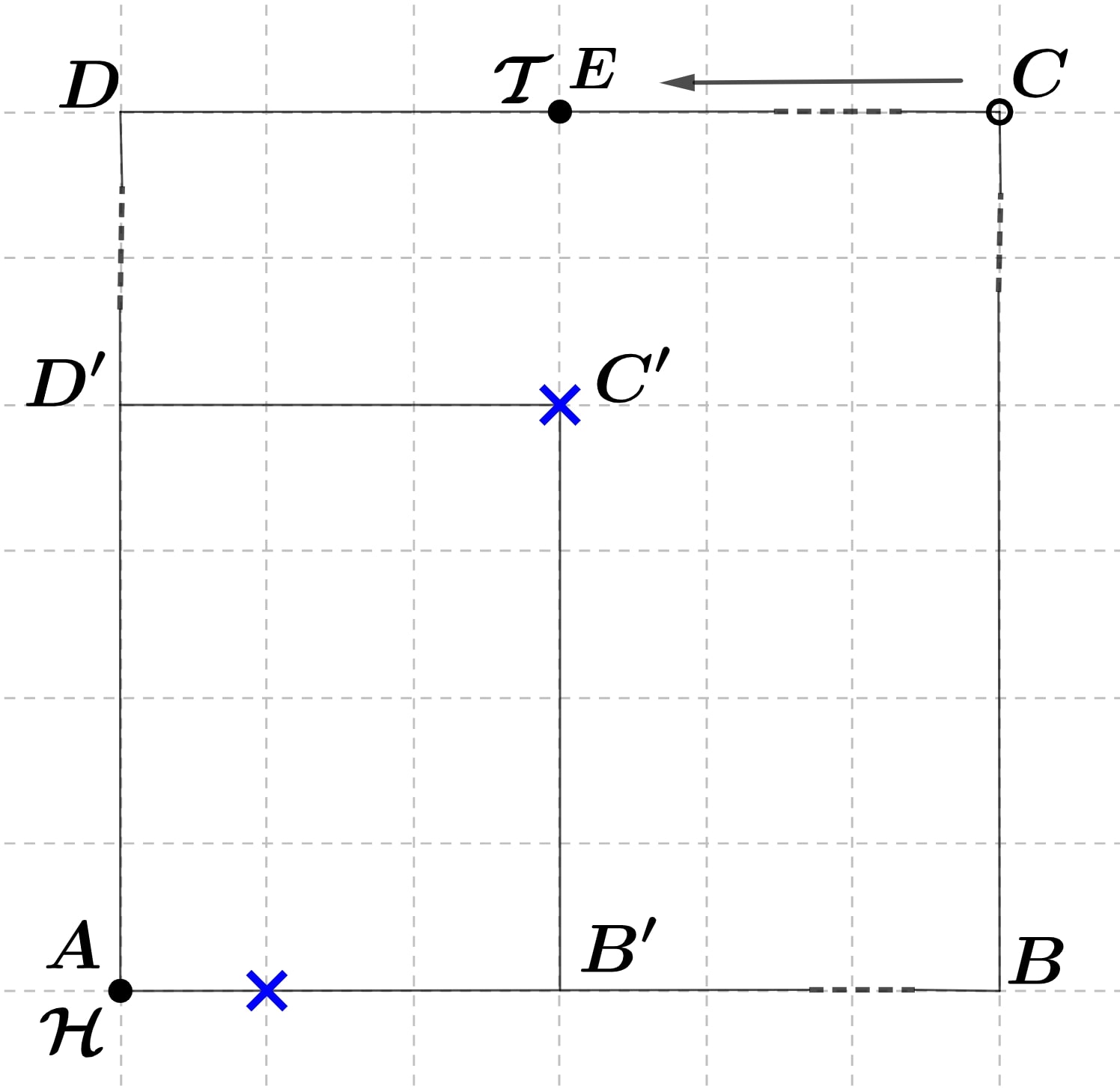}
     \caption{In phase 7 head move to it's fixed target and in phase 8 tail move to its target point.}
     \label{Fig:7}
    \end{figure}

\paragraph*{Phase 8:} Tail moves downwards up to $tail_{target}$. In this phase the position of tail will be upward of $tail_{target}$. So when in phase 8 tail will move downwards to the point $tail_{target}$. So when in phase 8 then $\neg S_0 \land S_1 \land S_2$ = true, but after tail move then $S_0$ is true.\\

\begin{theorem}
If we have a configuration $\mathcal{C}$ at some time t then after a finite number of moves by the tail $S_0$ becomes true.
\end{theorem}

\begin{proof}
In this phase by a finite number of moves by the tail robot, the arbitrary pattern given as input will form. After phase 6 may be the configuration has symmetry or not.
\begin{enumerate}
    \item Let the configuration is asymmetric. Then the tail robot is now on the CD side, where ABCD is the smallest enclosing rectangle. Let $AB'C'D'$ be the $s.rect$ of the target configuration. Then $tail_{target}$ will be in the downwards vertical line of the tail. So when the tail robot moves down to its target position AB will be always lexicographically larger string. 
    \item Let the configuration is asymmetric but the position of $tail_{target}$ is on the upper side or in its horizontal line or downside. This is only possible when the initial configuration is the same as the target without the tail's position. In this case, the tail will move to its position. No symmetry will occur during this move of the tail.
    \item Let the configuration is symmetric. As the initial configuration is asymmetric the symmetry may arise in phase 7, so when there is a vertical symmetry then let ABCD be the $s.rect$ and AB and BA be the larger string, as $S_4$ is true here so the target position for tail will be the downside of its recent position and after it moves and reaches to its $tail_{target}$ then $S_0$ is true.
\end{enumerate}
\end{proof}

\begin{figure}[ht]
    \centering
     \includegraphics[width=.55\linewidth]{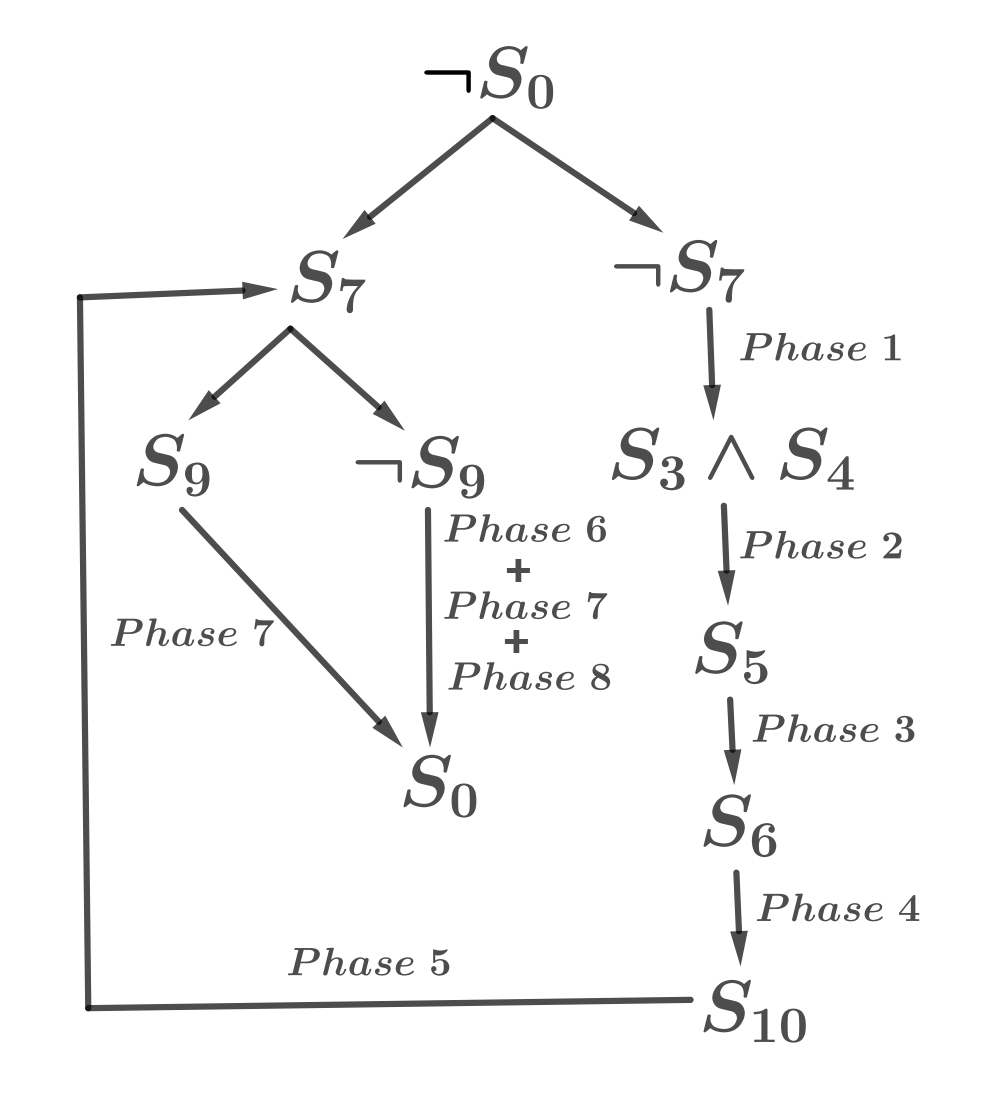}
     \caption{From $\neg S_0$ the algorithm terminates with $S_0$=true  }
     \label{Fig:8}
    \end{figure}
    
\subsection{Move Complexity of the algorithm}
In \cite{BoseAKS20} author proved that in an infinite grid arbitrary pattern formation by oblivious robots the optimal move is $O(\mathcal{D}')$ where $\mathcal{D}'=\max\{\mathcal{D},k\}$, where $\mathcal{D}$ is the side of the smallest square which can contain both initial and target configuration, that is, $D=\max\{AB,CD,A'B',C'D'\}$ and, $k$ is the number of robots. Here we show that in our algorithm the total number of required movements is $O(\mathcal{D}')$ which is asymptotically optimal. In our algorithm in phase 1, phase 2, and phase 3 only one robot moves. So here a robot uses maximum $O(\mathcal{D}')$ move. Then in phase 4 a robot will form a line and in phase 5 a robot reaches its position, so here also a robot has to move maximum $O(\mathcal{D}')$ steps. In phase 6, phase 7 and phase 8 only one robot will move, so here also maximum $O(\mathcal{D}')$ move is required. So we can conclude that in our algorithm total required movements are $O(\mathcal{D}')$, which is asymptotically optimal.

So starting from any asymmetric configuration it belongs to any of these eight phases, and our algorithm can form any arbitrary pattern by optimal move within finite time. Finally, we can conclude the theorem.

\begin{theorem}
Arbitrary pattern formation is solvable by optimal $O(\mathcal{D}')$ move in asynchronous scheduler by oblivious robots from any asymmetric configuration.
\end{theorem}

% \section{\textsc{Apf} algorithm with optimal time}
% In this section author gives an algorithm \texttt{FastAPF} which solves \textsc{Apf} problem in optimal time in an asynchronous scheduler. In this algorithm, robots are not oblivious. They have finite bit communicable (other robots can read the memory) persistent (does not get erased at the end of an LCM cycle) memory. Authors will refer to this memory as lights. While in a synchronous scheduler the time is measured by the number of rounds, in an asynchronous scheduler the time is measured in the unit epochs. An epoch is a time interval where all robots have been activated at least once. Let $\mathcal{D}'=\max\{k,\mathcal{D}\}$, then this section gives an algorithm which solves APF in $O(\mathcal{D}')$ epochs. Later in this section, it is shown that any algorithm solving APF takes $\Omega(\mathcal{D}')$ epochs.

\section{Optimal time Algorithm \texttt{FastAPF} for APF}
This section proposes an algorithm, named \textsc{FastAPF} for \textsc{Apf} which is time optimal. Before going to the algorithm, for convenience let go through some definitions, notions supporting the algorithm. Let $ABCD$ be the unique smallest rectangle enclosing a given initial configuration, where $AB\ge BC$. If $ABCD$ is not square then consider the set of strings $S$ to be $\{\lambda_{AB}, \lambda_{BA}, \lambda_{CD}, \lambda_{DC}\}$. If $ABCD$ is square then consider the set of strings $S$ to be $\{\lambda_{AB}, \lambda_{BA}, \lambda_{CD}, \lambda_{DC}, \lambda_{BC}, \lambda_{CB}, \lambda_{AD}, \lambda_{DA}\}$.

\subsection{Coordinate System setup}
If the given configuration is asymmetric then $S$ contains a lexicographically strictly largest string. Let $\lambda_{AB}$ be the largest among other strings in $S$. Then robots consider $A$ as the origin and $AB$ direction as the positive $x$ axis and $AD$ direction as the positive $y$ axis. In such a case, the first robot in $\lambda_{AB}$ string is said to be $head$. 

Each robot has a light. This light can take two colors, namely, \textsc{head} and \textsc{line} which are readable as well as communicable. The light can indicate another state when the light is off. We denote this one as \textsc{OFF} color.

Next suppose in a given configuration there is a robot, call it \textit{head}, with \textsc{head} color ON on the boundary of $ABCD$. There are two cases, firstly if the head is not situated in any corner and another when the head is at a corner of the $ABCD$, say $A$. For the first case we assume that the head is on the side $AB$, such that $AB\ge CD$.
For such a case consider $A$ as origin, $AB$ direction as positive $x$ axis, and $AD$ direction as positive $y$ axis. For the second case, consider the head robot is situated at a corner, say $A$. In such a case consider $A$ as the origin. If $AB>BC$ then consider $AB$ direction as positive $x$ axis and $AD$ direction as positive $y$ axis. If $AB=BC$, then we assume the configuration is asymmetric and hence there is a lexicographically strictly largest string, say $\lambda_{AB}$ is $S$. In such a case consider the $AB$ direction as a positive $x$ axis and $AD$ direction as a positive $y$ axis.

\paragraph*{Definition of Tail} ~\\
Case-I: (When \textsc{head} color is not ON) In this case we assume that the configuration is asymmetric. The last robot in the lexicographically strictly largest string in $S$ is said to be Tail.\\
Case-II: (When \textsc{head} color is ON) In this case the tail is the rightmost robot of the topmost horizontal line.

Let the $A'B'C'D'$ be the smallest rectangle enclosing the target pattern with $A'B'\ge B'C'$. Let $\lambda_{A'B'}$ be the largest (may not be unique) among all other strings in $S$ for target pattern. We denote the $i^{th}$ target position in $\lambda_{A'B'}$ string as $t_i$. Then the target pattern is to be formed such that $A=A'$, $A'B'$ direction is along the positive $x$ axis and $A'D'$ direction is along the positive $y$ axis.

Let $s=\max\{AD, A'D'\}$. Let name the lines parallel and above to $x$ axis by $H_1, H_1, \dots, H_{s}$. Name the vertical lines from left to right as $V_1,V_2,\dots$, where $V_1$ is the $y$-axis. Let at any time $C(t)$ the configuration be $C(t)$. Let in the target pattern the number of robots in $H_i$ line be $n_i$. Let in $C(t)$ the total number of robots below the line $H_i$ be $b_i$. Let in $C(t)$ the total number of robots above the line $H_i$ be $a_i$. Let $b'_i = \sum_{j<i}n_i$ and $a'_i = \sum_{j>i}n_i$. A horizontal line $H_i$ is said to be \textit{saturated} if $a_i=a_i'$ and $b_i=b_i'$ In Figure~\ref{sat} $H_5$ is a saturated line.

 \begin{figure}[ht]
    \centering
     \includegraphics[width=1\linewidth]{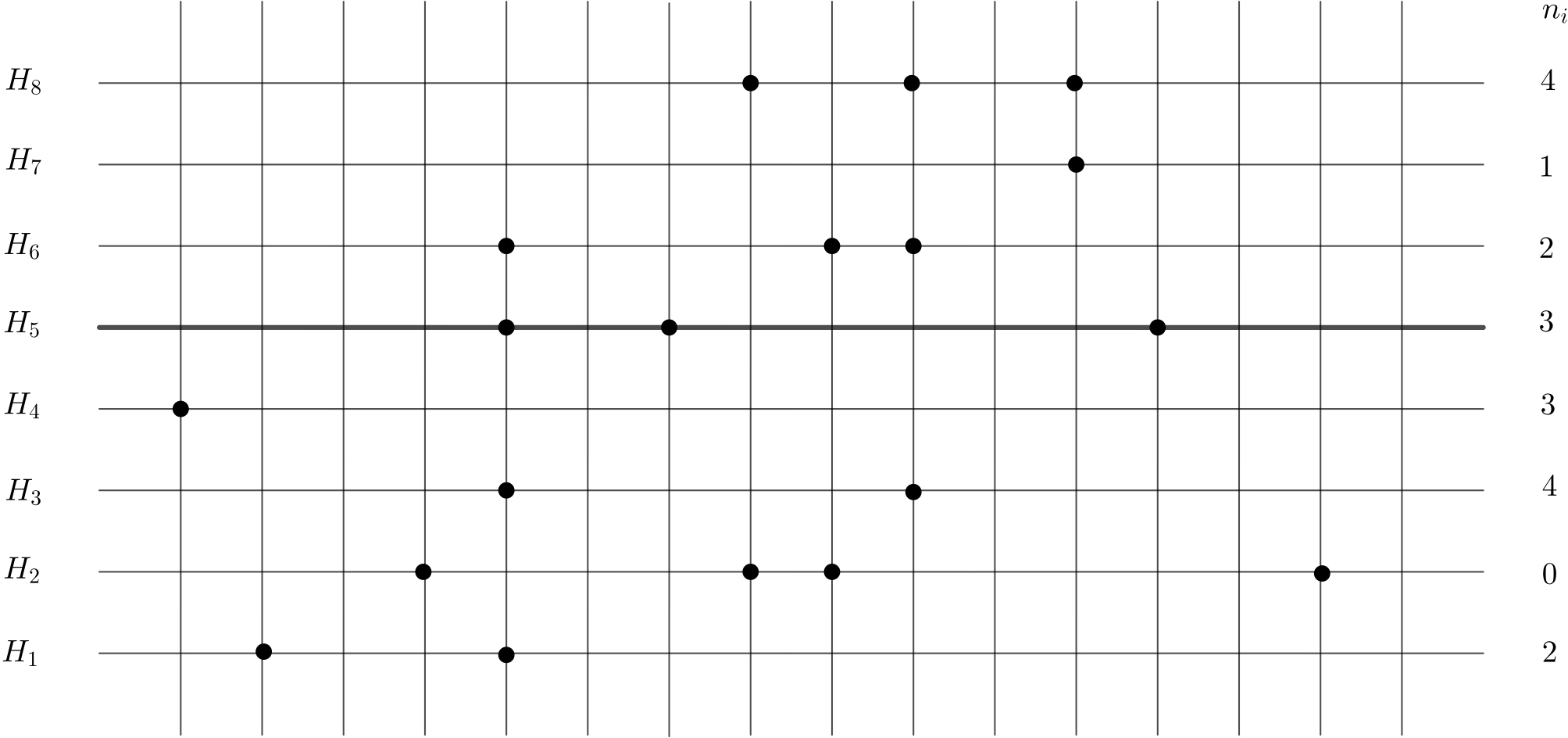}
     \caption{$H_5$ is a saturated line}
     \label{sat}
    \end{figure}

The next subsection describes all the intermediate procedures of \texttt{FastAPF}.
\subsection{Elements of the algorithm}

In order to optimize the time, our main motive is to make the algorithm so that it allows parallel movement of robot as much as possible avoiding collision. The algorithm is divided into six phases. Initially, since the configuration is asymmetric, a robot that activates at first can find out the things listed in Table \ref{tab:0}.
\begin{table}[ht!]
    \centering
    \begin{tabular}{|c|p{6cm}|}
    \hline
      1   &  Smallest enclosing rectangle $ABCD$ with $AB\ge BC$\\
      \hline
      2  & If the configuration is asymmetric,  lexicographically largest string $\lambda_{AB}$\\
      \hline
      3 & Head robot\\
      \hline
      4 & Tail robot\\
      \hline
    \end{tabular}
    \caption{}
    \label{tab:0}
\end{table}

  Throughout the algorithm, the head robot remains head using \textsc{head} color as the flag. And further, the coordinate system remains unaltered throughout the algorithm that has been taken care of. Once the \textsc{head} color is ON and hence the head robot is fixed. Any robot other than the head or tail is said to be \textit{inner} robot. Next, we define terminologies for different configurations.
  \begin{enumerate}
    \item $\mathcal{C}_{init}$ = Initial configuration
    \item $\mathcal{C}_{target}$ = Target configuration
    \item $\mathcal{C}$ = A possible configuration from initial configuration
    \item $\mathcal{C}'$ = $\mathcal{C}\setminus\{Head\}$
    \item $\mathcal{C}''$ = $\mathcal{C}\setminus\{Tail\}$
    \item $\mathcal{C}'''$ = $\mathcal{C}\setminus\{Head,Tail\}$
    \item $\mathcal{C}'_{target}$ = $\mathcal{C}_{target}\setminus\{Head\}$
    \item $\mathcal{C}''_{target}$ = $\mathcal{C}_{target}\setminus\{Tail\}$
    \item $\mathcal{C}'''_{target}$ = $\mathcal{C}_{target}\setminus\{Head,Tail\}$
  \end{enumerate}
  Next we define different conditions on configuration.
  \begin{enumerate}
      \item $C_0: \mathcal{C}=\mathcal{C}_{target}$
      \item $C_1: \mathcal{C}'= \mathcal{C}'_{target}$
      \item $C_2: \mathcal{C}''= \mathcal{C}''_{target}$
      \item $C_3: \mathcal{C}'''= \mathcal{C}'''_{target}$
      \item $C_4:$ \textsc{head} color is ON.
      \item $C_5:$ Head color is at corner.
      \item $C_6:$ All inner robots are with \textsc{line} color ON except those who are on a saturated line.
      \item $C_7:$ Tail is at a point with $x$-coordinate $\max\{AB+1,A'B'+1,k\}$ and $y$-coordinate $\max\{BC,B'C'\}$. 
   \end{enumerate}
\begin{enumerate}
    \item \textbf{Procedure-I:} 
    \textit{Input}: $\neg C_3\vee (C_3\wedge \neg C_1 \wedge \neg C_2)$
    
    In the first procedure, Head identifies itself and turns its \textsc{head} color on. Then it goes to origin if it is not.
    
    \textit{Output}: $C_4\wedge C_5$ is true.
    
    \paragraph*{Discussion} This procedure gets executed when either there is an inner robot not at its target position or all inner robots are at their target but neither head nor tail is at its target position. In such a case, the head robot first puts ON its \textsc{head} color and then moves leftwards until it reaches its origin. Note that if the input configuration is asymmetric then the configuration throughout the procedure remains asymmetric. And also the things listed in Table \ref{tab:0} remain unchanged.
    
    \item \textbf{Procedure-II:} 
    \textit{Input}: $ C_1\wedge \neg C_2  $
    
    In this procedure, the head moves to its target position and turns off its \textsc{head} color if it is ON.
    
    \textit{Output}: $C_0$ is true.
    
    \paragraph*{Discussion} This procedure gets executed when every robot except the head is at their respective target position. In this phase, the head occupies its target position and turns off its \textsc{head} color if it is ON. Now the head target must be on the $x$-axis, so either head needs to move right or left to reach its destination in this procedure. If the head needs to move left then clearly the things listed in Table~\ref{tab:0} remain unchanged. Also in the other case, since the head target is the first target position of the lexicographically largest string of the target pattern, the listed things in Table~\ref{tab:0} remain unchanged.
    
    \item \textbf{Procedure-III:} 
    \textit{Input}: $(C_2 \wedge \neg C_1 ) \vee (C_4 \wedge C_5 \wedge C_3)$
    
    \textit{Case-I:} If the $y$-coordinate of the tail target is the same as the $y$-coordinate of the tail, then the tail reaches at target by horizontal movements.
    
    \textit{Case-II:} If the $y$-coordinate of the tail target is not the same as the $y$-coordinate of the tail, then we consider the following cases.
    
    % \textit{Case-IIA:} If the \textsc{head} light is ON, then move downwards until the Case-I is achieved. 
    
    % \textit{Case-IIB:} If the headlight is OFF, then consider the following cases.
    
    \textit{Case-IIA:} If the $y$-coordinate of the tail target is greater than the $y$-coordinate of the tail, then move upwards until Case-I is achieved.
    
    \textit{Case-IIB:} If the $y$-coordinate of the tail target is less than the $y$-coordinate of the tail, move horizontally so that the $x$-coordinate of the tail target becomes the same as the $x$-coordinate of the tail. Then the tail moves downwards until condition $C_0$ is true.
    
    \textit{Output}: $C_0\vee C_1 $ is true.
    
    \paragraph*{Discussion} This procedure gets executed when either the pattern except the tail is formed ($C_2\wedge \neg C_1$) or all inner robots are at their target position and the head is at the origin with head color ON ($C_4 \wedge C_5 \wedge C_3$). For both cases, careful movement is assigned to the tail robot. For all the cases it can be checked that the things listed in Table~\ref{tab:0} remain unchanged. If the input configuration was $C_2\wedge \neg C_1$, after execution of this procedure the target pattern is formed ($C_0$). And if the input pattern was $C_4 \wedge C_5 \wedge C_3$ then after execution of this phase the resulting configuration is such that the target pattern is formed except head robot.    
    
    \item \textbf{Procedure-IV:} 
    \textit{Input}: $C_4\wedge C_5 \wedge \neg C_3 \wedge \neg C_7$
    
   In this procedure the tail moves right until its $x$-coordinate becomes $\max\{AB+1,A'B'+1,k\}$. Then the tail moves upward until its $y$-coordinate becomes $\max\{BC,B'C'\}$.
   
   \textit{Output}: $C_7$ is true.
   
   \paragraph*{Discussion} This procedure gets executed when at least one inner robot is not at its target position and the head robot is at its origin. In this procedure, the robot moves rightwards in order to ensure that there is enough big smallest rectangle of the current configuration for procedure-V and procedure-VI to execute and also to ensure that $AB>BC$ so that the coordinate system does not change.  
    
    \item \textbf{Procedure-V:} 
    \textit{Input}: $C_4\wedge C_5 \wedge \neg C_3 \wedge C_7 \wedge \neg C_6$
    
    In this procedure, if an inner robot is not on a saturated horizontal line then it checks whether its \textsc{line} color is ON or not. If the \textsc{line} color is not ON, then it counts the number of robots below it. Let's say the number is $b$. Then the robot counts the number, say, $l$ of robots on the left side to it in its horizontal line. Let $v=b+l+1$. Then the robot tries to move to the $v^{th}$ vertical line in its horizontal line. If $v^{th}$ vertical line is in its left (right) and its left (right) grid point is empty then it moves to left (right). When a robot reaches the $v^{th}$ vertical line, the robot turns on its \textsc{line} color.
    
    \textit{Output}: $C_6 $ is true (Figure~\ref{C6}).
    
    \paragraph*{Discussion} These procedures are executed when the head robot is at the origin with its \textsc{head} color ON and the tail robot's coordinate is $(\max\{AB+1,A'B'+1,k\},\max\{BC,B'C'\})$ and there is at least one inner robot which is not on the saturated line and with no line color ON. In this procedure, such robot reaches at $v^{th}$ vertical line and turns ON its \textsc{line} color. At the end of this procedure in the configuration, the head and tail remain at their starting position and each the inner robot is either with \textsc{line} color ON or on a saturated line.

     \begin{figure}[ht]
    \centering
     \includegraphics[width=1\linewidth]{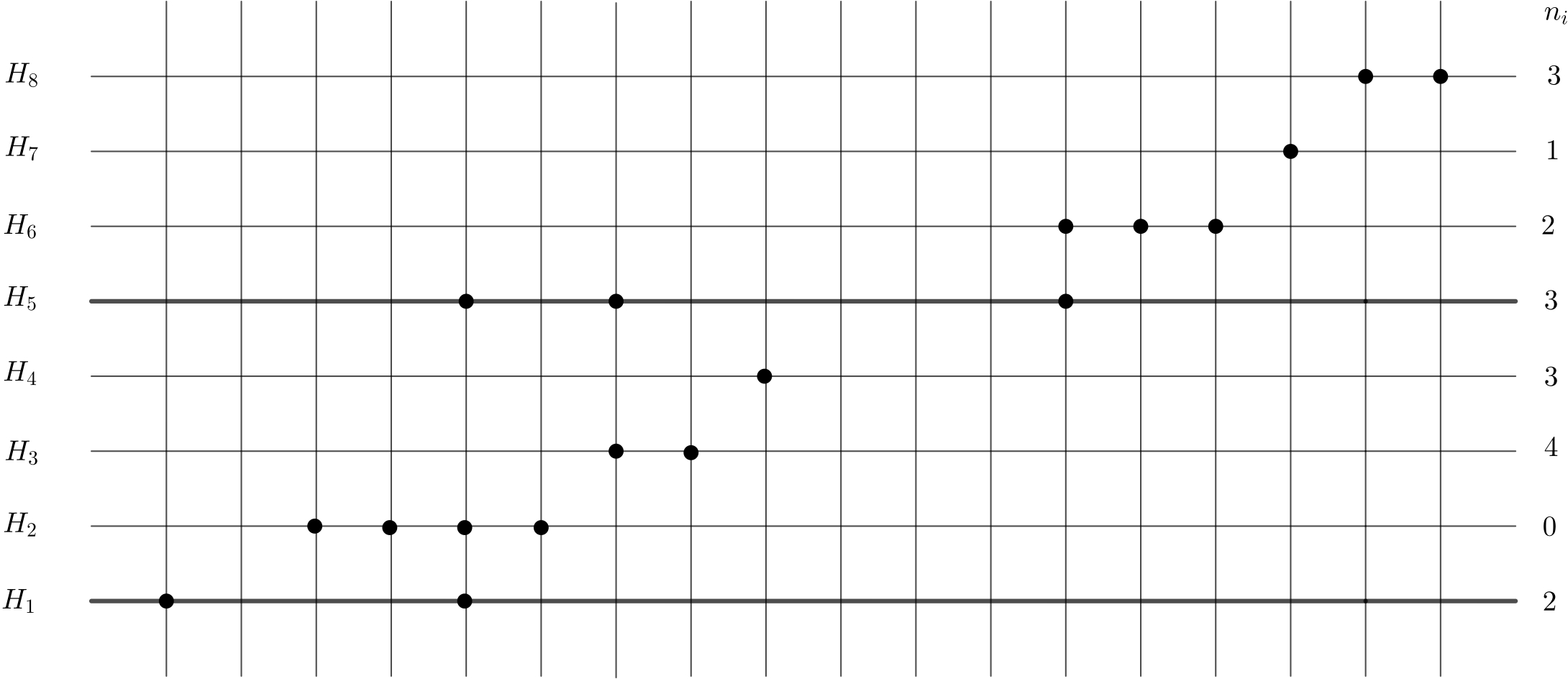}
     \caption{When $C_6$ true}
     \label{C6}
    \end{figure}
     \item \textbf{Procedure-VI:} 
    \textit{Input}: $C_4\wedge C_5 \wedge \neg C_3 \wedge C_7 \wedge C_6$
    
    In this procedure, there are two exhaustive cases.
    
    \textbf{Case-I:}  
     This is the case when the robot sees that its own horizontal line is not saturated. In such a case, the robot counts its vertical line number. Let the robot be at $v_i^{th}$ vertical line. Then it moves to the horizontal line which contains the $v_i^{th}$ target position by vertical movements until it reaches there (Figure~\ref{phase6}).
     
    \textbf{Case-II:} This is the case when the robot sees that its own horizontal line is saturated. In this case firstly if its \textsc{line} color is ON, then it turns it OFF. Otherwise, if the robot is the $k^{th}$ robot on its horizontal line from the left, then it tries to reach the $k^{th}$ target position on that horizontal line. If $k^{th}$ target position in its horizontal line is in its left (right) and its left (right) grid point is empty then it moves to left (right).
    
    \textit{Output}: $C_3 $ is true.
    
    \paragraph*{Discussion} In the input configuration of this procedure an inner robot is either on a saturated line or not on a saturated line but with its \textsc{line} color ON. In this procedure, no two inner robots with \textsc{line} color ON are on the same vertical line. In this phase, two types of movements are happening simultaneously. Firstly robots with line color ON but not on a saturated line do vertical movements and rests do horizontal movements. Eventually, all robots will reach their destined horizontal line and all horizontal lines will become saturated. Next, we need to show no two robots collide in this procedure. We can guarantee that if no robot with \textsc{line} reaches the already saturated line. Now from the definition of a saturated line, one can note that every robot below a saturated line has its target horizontal line below the saturated line. And also every robot above a saturated line has its target horizontal line above the saturated line. Hence no robot making vertical movement will reach a saturated line where horizontal movement is possibly happening. Hence this procedure is collision-free.   
    
     \begin{figure}[ht!]
    \centering
     \includegraphics[width=1\linewidth]{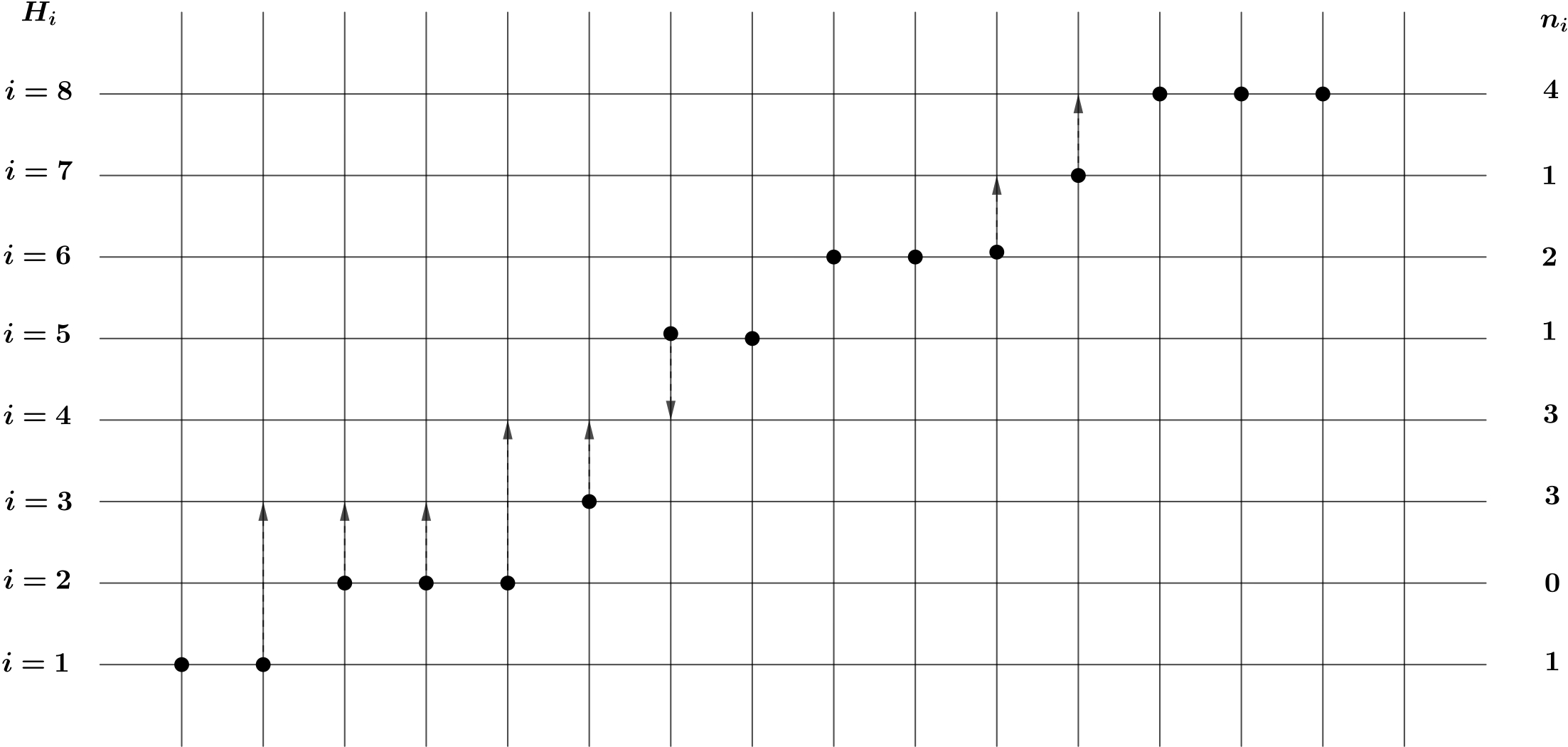}
     \caption{Procedure VI illustration}
     \label{phase6}
    \end{figure}
    \end{enumerate}
In next subsection formally presents the algorithm \texttt{FastAPF} in flow chart form. 

\subsection{Algorithm FastAPF}
The main algorithm \textsc{FastAPF} is presented below in flow chart in Figure~\ref{algo:2}.
\begin{figure}[ht!]
    \centering
     \includegraphics[width=1\linewidth]{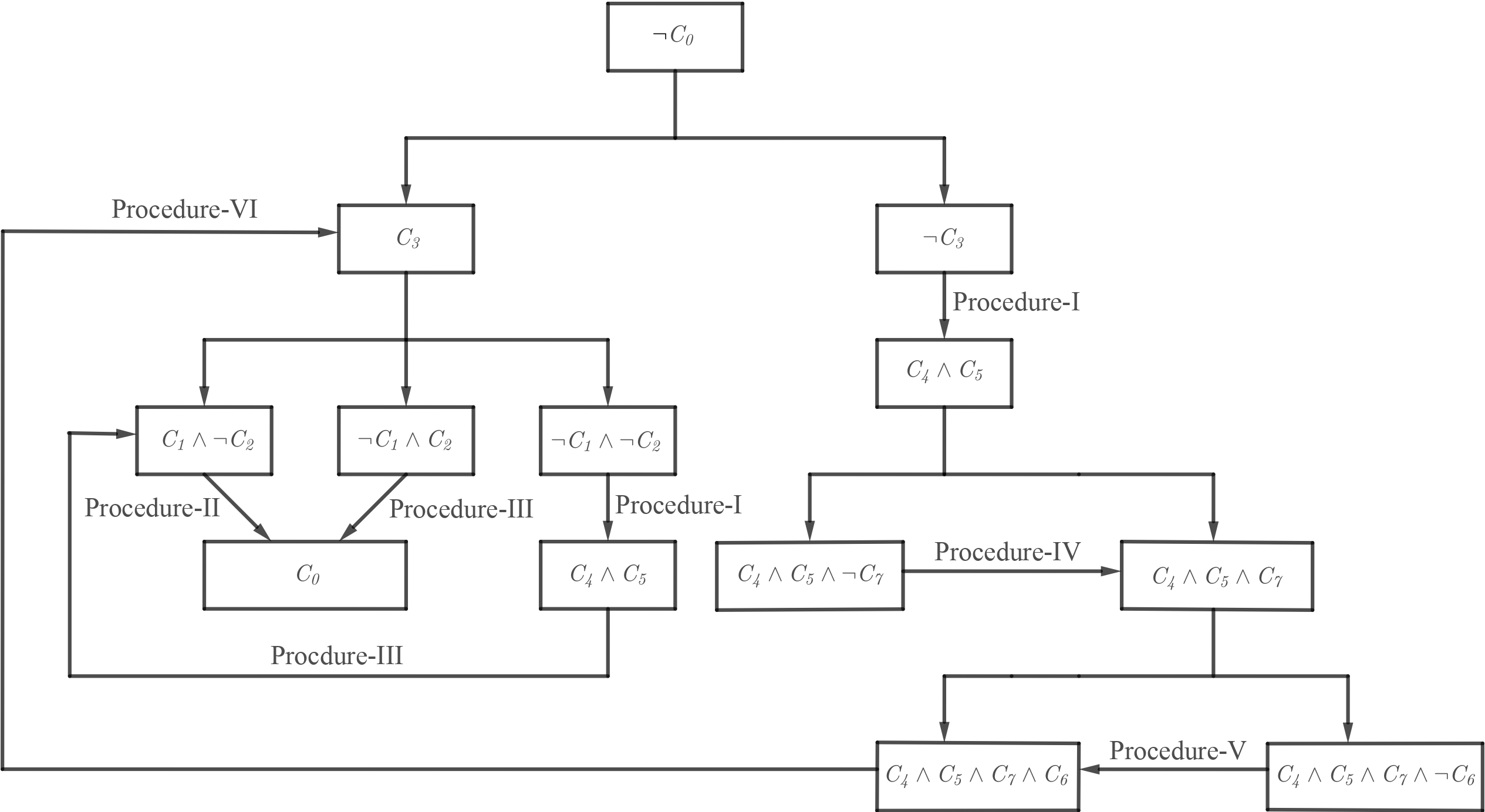}
     \caption{Algorithm \textsc{FastAPF} flow chart}
     \label{algo:2}
    \end{figure}

Starting from any possible configuration where $C_0$ is not true, in the flow chart, we can observe that the path ends to the $C_0$ configuration passing through some procedures. Hence we conclude the following.
\begin{theorem}
The \textsc{FastAPF} algorithm solves the \textsc{APF} problem in $\mathcal{LUMI}$ model.
\end{theorem}
\subsection{Time Complexity of FastAPF algorithm}
We show that each procedure included in our algorithm takes $O(\mathcal{D}')$ epoch. Since the algorithm \textsc{FastAPF} is composition of six mentioned procedures, this will prove our claim that, the algorithm \textsc{FastAPF} solves \textsc{APF} problem in $O(\mathcal{D}')$ epochs.

In Procedure-I, Procedure-II, Procedure-III, and Procedure-IV only one robot is making its movement, so all the phases can take maximum $\mathcal{D}'$ epochs to complete. In Procedure-V in each horizontal line, robots are making horizontal movements parallelly. On a horizontal line $H_i$, for all robots to reach the destined vertical line maximum it will take 2$\mathcal{D}'$ epochs. Hence Procedure-V takes at most 2$\mathcal{D}'$ epochs to complete. In procedure-VI, a robot with \textsc{line} color on makes vertical movements. All vertical movements in this procedure happen simultaneously. So all the vertical movements are done in $\max\{BC,B'C'\}$ epochs. Then in a saturated line horizontal movements happen in order to reach the target positions. Similarly, this also takes a maximum 2$\mathcal{D}'$ epochs. Hence this procedure also takes $O(\mathcal{D}')$ epochs. Hence we conclude this discussion with the following theorem.

\begin{theorem}
The \textsc{FastAPF} algorithm solves the \textsc{APF} problem in $\mathcal{LUMI}$ model in $O(\mathcal{D}')$ epochs.
\end{theorem}

\begin{figure}[ht]
    \centering
     \includegraphics[width=1\linewidth]{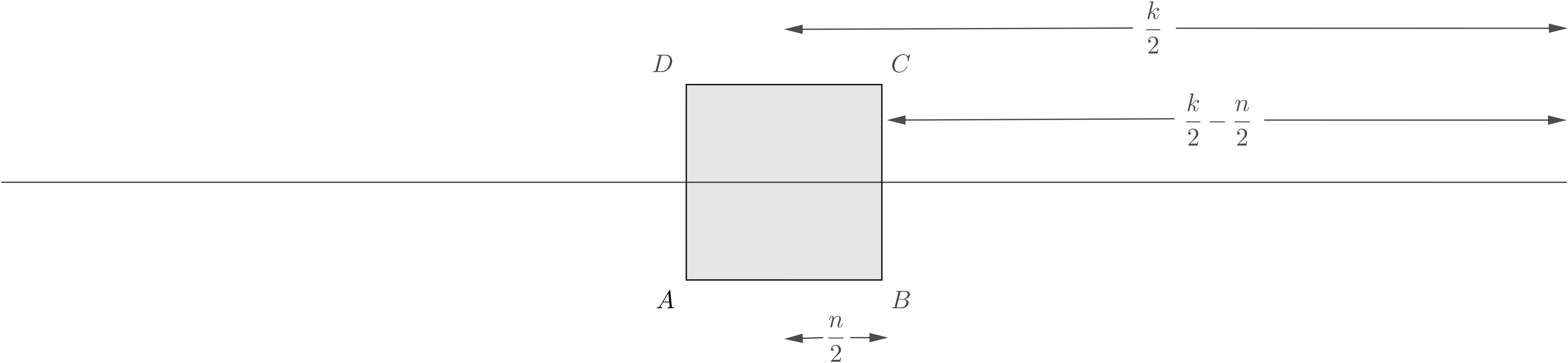}
     \caption{An image related to Theorem \ref{th10} }
     \label{optimal0}
    \end{figure}

Now we show that our algorithm is asymptotically optimal in time. Let's consider an initial configuration of robots has the smallest enclosing rectangle a square of length $n$ and each grid point of the square has a robot. Thus there are $k=n^2$ robots in total. Let the target pattern be a compact line. Then we can see that the farthest point on the line (wherever it may be placed) from the initial configuration is at least $\dfrac{k-n}2=\dfrac k 2-\dfrac{\sqrt{k}}2$ (See Figure~\ref{optimal0}). Hence there is a robot that at least needs to move $\dfrac k 2-\dfrac{\sqrt{k}}2$ steps. Hence at least $\dfrac k 2-\dfrac{\sqrt{k}}2$ epochs is necessary. In this case, note that $\mathcal{D}=k$. Hence we can state the following.

\begin{theorem} \label{th10}
Any algorithm solving \textsc{Apf} problem requires $\Omega(\mathcal{D}')$ epochs.
\end{theorem}

This above result shows that the algorithm \textsc{FastAPF} is time optimal asymptotically.

\section{Conclusion}
This paper studied both move and time-optimal algorithms for \textsc{arbitrary pattern formation} (\textsc{Apf}) problem on an infinite grid by robots starting from any asymmetric initial configuration. Here this work gives an algorithm for the \textsc{Apf} problem in $\mathcal{OBLOT}$ model which is asymptotically move optimal. Then this work proposed another algorithm for the \textsc{Apf} problem in $\mathcal{LUMI}$ model which is asymptotically time optimal. If $\mathcal{D'}$ is the maximum of the number of robots and the side of the smallest enclosing square enclosing both target and initial configuration, then the first algorithm uses total $O(\mathcal{D}')$ movements and the second algorithm takes total $O(\mathcal{D}')$ epochs to solve APF problem. Here we assume the initial configuration is asymmetric but for further work, it will be interesting to study any fast algorithm when the configuration is not asymmetric initially. One can consider another problem where the initial or the target configuration of robots has multiplicity points.

\section{References}
\bibliographystyle{tfnlm}
\bibliography{interactnlmsample}

\end{document}